\documentclass[letterpaper, 10pt, journal, twoside]{IEEEtran}
\IEEEoverridecommandlockouts                            
\usepackage[dvipdfmx]{graphicx}
\usepackage{amsmath, amssymb, bm, mathrsfs, url,amsfonts}
\usepackage{cite, color, graphicx, subcaption, comment}
\usepackage{algorithm,algpseudocode,color, mathtools}
\usepackage[outdir=./]{epstopdf}
\mathtoolsset{showonlyrefs=true,showmanualtags=true}
\usepackage[colorlinks,bookmarksopen,bookmarksnumbered,citecolor=red,urlcolor=red]{hyperref}
\usepackage{booktabs}

\algrenewcommand\algorithmicrequire{\textbf{Input:}}
\algrenewcommand\algorithmicensure{\textbf{Output:}}
\def\BibTeX{{\rm B\kern-.05em{\sc i\kern-.025em b}\kern-.08em
		T\kern-.1667em\lower.7ex\hbox{E}\kern-.125emX}}
\newtheorem{theorem}{Theorem}[section]
\newtheorem{proposition}[theorem]{Proposition}
\newtheorem{definition}[theorem]{Definition}

\newtheorem{lemma}[theorem]{Lemma}
\newtheorem{remark}[theorem]{Remark}

\DeclareMathOperator*{\argmin}{arg\,min}
\newcommand{\Lap}{\mathop{\rm Lap}}

\newcommand{\sat}{\mathop{\rm sat}}

\newcommand{\sol}{\mathop{\rm sol}}

\def\rnum#1{\expandafter{\romannumeral #1}}

\title{
	Cell Zooming with Masked Data for Off-Grid Small Cell Networks: 
	Distributed Optimization Approach
}

\author{Masashi~Wakaiki,~Katsuya~Suto,~and~Izumi~Masubuchi
	\thanks{This work was supported by Japan Society for the Promotion of Science (JSPS) 
		KAKENHI Grant Number
		JP17K14699.}
	\thanks{M. Wakaiki and I. Masubuchi are with the Graduate School of System Informatics, Kobe University, Kobe, 
		6578501
		Japan, e-mail: wakaiki@ruby.kobe-u.ac.jp; msb@harbor.kobe-u.ac.jp.}
	\thanks{K. Suto is with the Graduate School of Informatics and Engineering, The University of Electro-Communications, Tokyo, 1828585 Japan, email: k.suto@uec.ac.jp.}		}

\begin{document}

	\maketitle
	\thispagestyle{empty}
	\pagestyle{empty}

\begin{abstract}
	Cell zooming has been becoming an essential enabler 
for off-grid small cell networks. Traditional models often 
utilize the numbers of active users in order to
determine cell zooming strategies. However,  
such confidential measurement data
must be concealed from others. 
We therefore propose a novel cell zooming 
method with masking noise. The proposed 
algorithm is designed based on  distributed optimization, 
in which each SBS locally solves
a divided optimization problem and learns 
how much a global constraint is satisfied or violated 
for temporal solutions. 
The important feature of this distributed control method is
robustness against masking noise. 
We analyze the trade-off between 
confidentiality and optimization accuracy, using the notion of differential privacy. 
Numerical simulations show that the proposed distributed 
control method outperforms a standard centralized control 
method in the presence of masking noise.
\end{abstract}


\section{Introduction}
Off-grid small cell base stations (SBSs) are expected to be a promising technology that can boost the capability of traditional macro cell base stations (MBSs) with minimum capital expenditures 
and operating expenses~\cite{Couto2018,HossainAccess2020}.
Off-grid SBSs achieve sustainable operation without power grids by employing energy harvesting modules such as solar photovoltaics.
LG U+ and China Mobile deployed off-grid base stations (BSs) 
and started the communications services~\cite{KoIoTJ2019}.
The study~\cite{Piro2013} reports that a solar panel with an area 
of $0.6\,\text{m}^2$ can harvest up to $500\,\mathrm{W}$ with a conversion efficiency of $14\,\%$.
Although the level of energy harvesting is high enough to operate SBSs,
energy harvesting conditions strongly depend on time and location.
Therefore, it is important 
to control the operation of off-grid SBSs according to energy harvesting conditions. 

As a fundamental problem of off-grid SBSs, cell zooming 
from the operator's point of view 
has been under intensive investigation for the past several years.
Previous studies on cell zooming can be classified into a 
\textit{physical-layer solution}~\cite{An2014,Chang2014ICC,Maghsudi2017TWC,Maghsudi2017TGCN,Liu2018TWC,KoIoTJ2019}
and a \textit{network-layer solution}~\cite{Zhang2016Energy,Jiang2018,Wakaiki2018EHSCN,Chamola2018, Lee2017TWC, AlqasirTGCN2020, MiozzoTGCN2020, Suto2018Letter, Wu2018, KuTVT2020, PiovesanTGCN2021} 
based on their objectives.

The physical-layer solution aims at maximizing physical-layer wireless 
capability such as capacity and energy efficiency by using channel 
state information and battery state information. In this approach, 
the network operation is optimized without taking user 
demands into account. An~\textit{et~al.} \cite{An2014} study the impact of SBS density on 
the energy efficiency in cellular networks by
the stochastic geometry 
theory and propose an optimal user association strategy based on quantum 
particle swarm optimization. The capacity
of energy harvesting small cells is maximized via
joint sleep-wake scheduling and transmission power control
in~\cite{Chang2014ICC}.
Maghsudi and Hossain \cite{Maghsudi2017TWC} investigate a distributed 
user association problem of downlink small cell networks, where SBSs 
select  users by maximizing their successful transmission 
probability. In~\cite{Maghsudi2017TGCN}, 
a small cell network
is modeled as a competitive market with uncertainty
so as to develop
a distributed user association scheme for energy harvesting
small cell networks under uncertainty.
Liu~\textit{et~al.}~\cite{Liu2018TWC} develop a mathematical model to assess the association probabilities and coverage probabilities of non-uniform off-grid small cell networks with cell zooming.
Ko~\textit{et~al.}~\cite{KoIoTJ2019} propose a novel architecture of off-grid small cell networks where SBSs employ wireless power transfer to share their harvested energy with the terminals of accommodated users.

The objectives of the network-layer solution are to guarantee 
the quality of services such as delay and throughput and to minimize 
an operation cost. In addition to channel 
state information and battery state information, this approach 
exploits traffic load, user distribution, service types, and so on. In other 
words, the network-layer solution is based on the 
dynamics of user's behaviors as well as channel 
state information and battery state information.
Zhang~\textit{et~al.}~\cite{Zhang2016Energy} propose an energy-efficient cell 
zooming method, in which user association, SBSs activation, and 
power control are jointly optimized based on battery state information and user demands. 
The performance of network-layer cell zooming depends 
on the accuracy of traffic load estimation. 
To overcome this difficulty,
an estimation 
scheme based on various kinds of big data is proposed in~\cite{Jiang2018}. 
Model predictive control is applied to an online cell 
zooming problem in~\cite{Wakaiki2018EHSCN}, which allows
small cell networks to provide sustainable communication to users with optimal energy efficiency. 
Chamola~\textit{et~al.}~\cite{Chamola2018} 
demonstrate that energy efficiency and communication 
latency can be improved by a cell zooming technique in a real BS deployment scenario.
A recent trend in network-layer solutions is to design a fast algorithm to solve a non-convex optimization problem of  cell zooming  by using the ski rental problem~\cite{Lee2017TWC}, network centrality~\cite{AlqasirTGCN2020}, and layered learning~\cite{MiozzoTGCN2020}.
Another recent trend is to develop a novel paradigm of off-grid small cell networks, e.g., wireless mesh networks consisting of off-gird SBSs~\cite{Suto2018Letter}, content-aware cell zooming~\cite{Wu2018}, vehicular edge computing~\cite{KuTVT2020}, and energy sharing among SBSs~\cite{PiovesanTGCN2021}.

One of the main drawbacks of the existing studies above is that they do not consider
the security of control systems.
To determine a cell zooming policy, a central controller, which is mostly contained in an MBS,
collects  the numbers of active users in the service areas of SBSs.
Such social data can be also used for
commercial and administrative purposes, e.g., 
to find socioeconomic trends, to design public transportation systems, and
to analyze people-flow after catastrophic events such as earthquakes for
the mitigation of secondary disasters \cite{Sekimoto2011, Toch2019}.
To steal the confidential measurement data, 
intruders may
do packet sniffing (see, e.g., \cite{Ansari2003,Sikos2020}), i.e., 
analyze control packets transmitted from SBSs.
To prevent security threats and
commercial losses,
mobile network providers need to carefully 
protect the confidentiality of the measurement data.
Adding masking noise to the measurement data is a simple method to
enhance the security of cell networks but
simultaneously worsens
the performance of cell zooming.
To the
best of our knowledge, there is no literature on 
noise-resistant cell zooming techniques.

This research is classified into a network-layer solution. 
We formulate the minimization problem of cell zooming for off-grid SBSs
and assume that 
intruders can analyze control packets but have no
knowledge of controllers or data correlation. 
In the minimization problem,
two elements are evaluated: the user association by the SBSs and the available energy of batteries in the SBSs. 
Since
the SBSs are powered by energy harvesting devices, 
the cell network becomes more ``green'' as the number
of users associated with the SBSs increases.
On the other hand, 
the depletion of the energy in the SBSs should be avoided
for the sustainable operation of the cell network.
For these reasons, we consider the above two elements in the 
minimization problem of cell zooming.

{\em Distributed optimization} plays a key role 
in the proposed cell zooming method. 
Each SBS locally solves a divided optimization problem, and the central controller
only checks whether or not 
the temporal control policies computed by the SBSs
achieve the global constraint, i.e.,  the full accommodation of active users.
The SBSs know only local information but learn 
the extent to which the global constraint  is satisfied or violated for 
the temporal control policies.
The essential idea of distributed optimization we use  can be found in \cite[Chapter~5]{Bertsekas1999}.
We refer the reader to \cite{Yang2010, Yang2020} for the recent developments of distributed optimization.

Reduction of computational complexity is an important issue in the distributed method, because
each SBS, which generally has limited computational resources, solves a divided optimization problem iteratively.
However, the minimization 
problem originally formulated for cell zooming is equivalent to
a mixed integer nonlinear programing problem due to
the discrete behavior of  sleep-wake scheduling.
To reduce computational complexity,
we apply the $\ell^1$ convex relaxation  \cite{Donoho2006}  and 
analyze the divided optimization problem each SBS solves.

The main contributions of this paper are as follows:
\begin{itemize}
\item 
We develop a distributed cell zooming method that is robust
to masking noise.
In the conventional centralized control method, 
both the objective function and the constraint
of the minimization problem are negatively affected by masking noise. In contrast,
only the constraint is disturbed by masking noise in the distributed control method.
Therefore, the proposed method is less 
susceptible to masking noise
than the conventional centralized method, which is also 
verified from numerical simulations.

\item 
As a by-product of the distributed control approach,
we can  simultaneously determine the transmission powers and 
the sleep-wake schedules of the SBSs. Moreover,
the distributed approach can deal with 
the situation where the user densities in
the areas of the SBSs are different.
The centralized cell zooming method based on model predictive control \cite{Wakaiki2018EHSCN} does not have these features.

\item
We propose a computationally-efficient algorithm to solve the 
minimization problem of cell zooming.
The proposed algorithm finds an approximate solution with low computational complexity.
Numerical results show that the approximation error is quite small. 

\item
We analyze the trade-off between confidentiality  and optimization accuracy,
by using the notion of differential privacy \cite{Dwork2014} 
as a measure of confidentiality.
This trade-off analysis 
can be used as a guideline to determine the noise intensity
for cell zooming with masked data.
\end{itemize}

The rest of this paper is organized as follows.
In Section~\ref{sec:system_model}, we formulate the minimization problem
of cell zooming and then compare the effects of masking noise between the
centralized and distributed control methods. In
Section~\ref{sec:Distributed_cell_zooming}, 
we develop a distributed algorithm for cell zooming with masked data. 
Section~\ref{sec:differential_privacy} is devoted to analyzing
the trade-off between confidentiality and optimization accuracy.
In~Section~\ref{sec:numerical_simulation}, we evaluate the performance 
of the distributed control method, using numerical examples.
Section~\ref{sec:conclusion} concludes the paper.

The proposed method is based on the technique developed in 
our conference paper \cite{WakaikiVTC2019}. 
The new parts in the journal version are as follows.
First, we provide numerical simulations to show 
how resistant the distributed control
method is to masking noise.
Second, we significantly reduce  the computational complexity of
the distributed control method by developing an explicit formula
for an approximate solution of the divided optimization problem.
Finally, the analysis of the
trade-off between confidentiality and optimization accuracy is
completely new.

\section{System Model and Problem Formulation}
\label{sec:system_model}

In this section, we first describe the 
power flow and the user association of off-grid SBSs.
Next we formulate the minimization problem of cell zooming for SBSs.
Finally, we compare centralized and distributed control approaches 
from the viewpoint of robustness against masking noise.
The main notations used in this article are summarized in Table~\ref{tab:ListNotation}.
\subsection{Power flow of SBS}
Let us consider that
$N$ SBSs  are deployed within 
the coverage region of an MBS powered by the grid.
The SBSs are off-grid and equipped with an energy harvesting device and 
a battery for energy storage.
We denote by $w_i[k]$ the harvested power [W]
of the $i$th SBS at time $k$.
The SBSs consume the energy of the battery for the transmission power 
and the system power.
Let $u_i[k]$ and $s_i[k]$ denote 
the transmission power [W] and
the system power [W] of the $i$th SBS at time $k$, respectively. 
Then the power consumption of the $i$th SBS at time $k$ is given by
$u_i[k]/\gamma  + s_i[k]$,
where a positive quantity $\gamma$ is a power amplifier efficiency.
The system power $s_i[k]$ takes two values, depending on
whether the SBS is in an active mode or a sleep mode:
\begin{equation}
s_i[k] := 
\begin{cases}
s_{\text{active}} & \text{if $u_i[k] > 0$ ~($i$th SBS is active)}, \\
s_{\text{sleep}} & \text{if $u_i[k] = 0$ ~($i$th SBS is sleep)}.
\end{cases}
\end{equation}
Notice that the SBS is in the sleep mode
if and only if the transmission power is  set to $0$.

The power flow of the battery in the $i$th SBS at time $k$ (positive in charging)
is described by
$w_i[k]-u_i[k]/\gamma  - s_i[k].$
Let 
$x_i[k] \in [0,X_{\max}]$ be
the  residual energy [J] of the $i$th SBS at time $k$.
The  energy  $x_i[k]$ has the dynamics:
\begin{align}
x_i[k+1] = 
\sat  \Big( 
x_i[k] + h\big(w_i[k]-u_i[k]/\gamma - s_i[k]\big) \Big),
\end{align}
where a positive quantity $h$ is a sampling period [s] and
a function $\sat: \mathbb{R} \to [0,X_{\max}]$ is the saturation function defined by
\begin{equation}
\label{eq:sat_def}
\sat (x) :=
\begin{cases}
0 & \text{if $x < 0$}, \\
x & \text{if $x \in [0,X_{\max}]$}, \\
X_{\max} & \text{if $x > X_{\max}$}. \\
\end{cases}
\end{equation}
Fig.~\ref{fig:power_flow} shows the power flow of the SBSs.

\begin{table}[tb]
	\centering
	\caption{List of main notations.}
	\label{tab:ListNotation}
	\begin{tabular}{c|l} \hline
		\textbf{Notation} & \textbf{Definition}   \\ \hline\
		$N$ & Number of SBSs \\
		$x_i$ & Residual energy of $i$th SBS   \\
		$u_i$ & Transmission power of $i$th SBS \\
		$w_i$ & Harvested power of $i$th SBS \\
		$s_i$ & System power of $i$th SBS   \\
		$v_i$ & Masking noise for $i$th SBS  \\
		$s_{\rm active}$ & System power in active mode \\
		$s_{\rm sleep}$ & System power in sleep mode \\
		$X_{\max}$ & Maximum capacity of battery \\ 
		$\gamma$ & Amplifier efficiency \\
		$h$ & Sampling period \\
		$\sat$ &  Saturation function: $\mathbb{R} \to [0,X_{\max}]$ \\
		$U_{\rm Macro}$ & Maximum capacity of MBS \\
		$U_i$ & Number of users in $i$th SBS area \\
		$A$ & Size of SBS coverage area\\
		$F$ & Number of users served by one SBS \\
		$\bf u$ & Vector of transmission powers \\
		$\bf U$ & Vector of numbers of users  \\
		$\bf v$ & Vector of masking noise \\
		$P_k$ & Objective function at time $k$  \\
		$\lambda$ & Weighting parameter of $P_k$\\
		$\textrm{P}$ & Probability \\
		$\textrm{E}$ & Expectation \\
		$\Lap(\rho)$ & Laplace distribution with parameter $\rho$ \\
		$\|u\|_i$ & $\ell^i$ norm of $u$ ($i=0,1,2$) \\
		\hline	
	\end{tabular}
\end{table}

\begin{figure}[tb]
	\centering
	\includegraphics[width = 7.5cm]{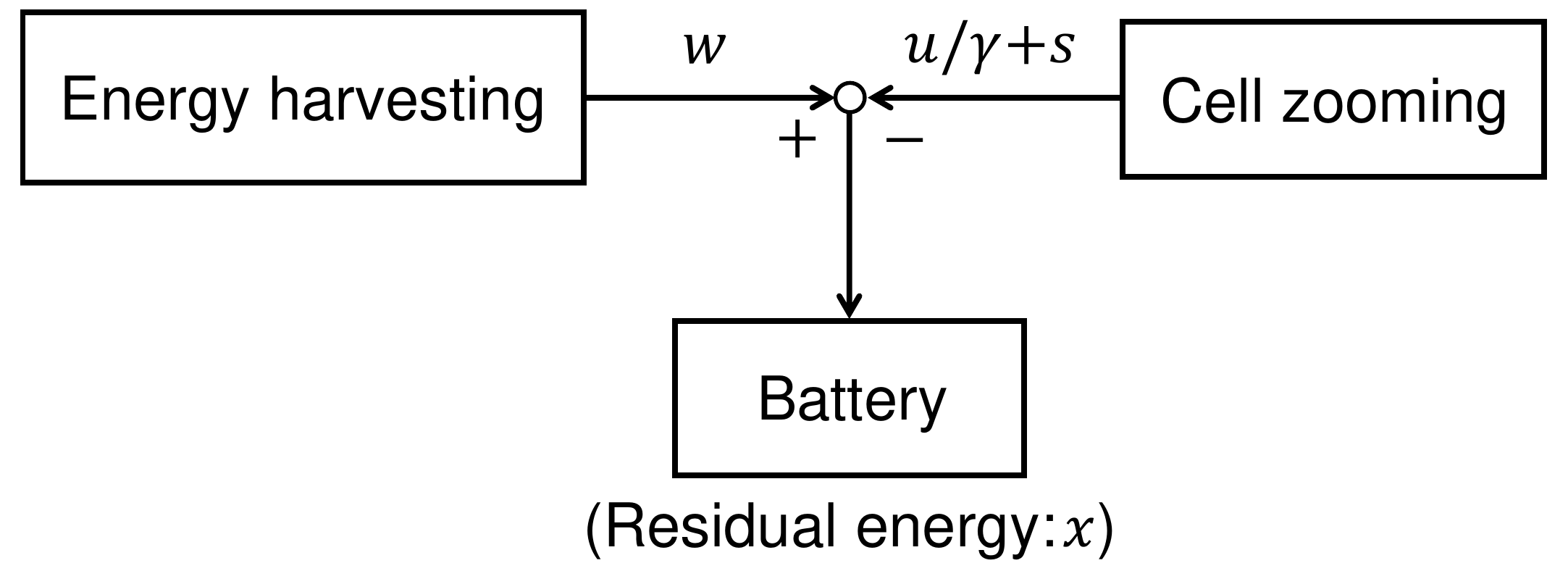}
	\caption{Power flow of SBS.}
	\label{fig:power_flow}
	\vspace{-10pt}
\end{figure}

\subsection{User association}
Let $U_{\text{Macro}}$ be the maximum number of users associated with the MBS
and $U_i[k]$ be
the number of 
users who request communication services in the coverage area of the $i$th SBS at time $k$.
Let $A$ be the size [$\text{km}^2$] of the coverage area of each SBS.  We assume that 
active users are uniformly distributed in the coverage areas. 
Let $F(u,U)$ denote 
the number of users
who an SBS can accommodate by the transmission power $u$ in
its service area with $U$ users.
In \cite{Suto2017}, a formula for $F(u,U)$ is given by
\begin{equation}
F(u,U) := r U u^{\frac{10}{19}},
\end{equation}
where a positive quantity $r$ is given by 
$
r := \pi 10^{b+\frac{30}{19}} /A
$
with
\begin{equation}
b := -\frac{1}{19}
\left( \frac{5}{2}
\mathrm{erfc}^{-1}\left(\frac{8}{3}\mathrm{BER}(Q)\right)^2+\sigma+Z
\right).
\end{equation}
In this definition,
$\sigma$, $Z$, and $\mathrm{BER}(Q)$ are
the system noise  [dBm] at a user device, 
the path loss factor  [dBm] other than distance related factor in the Walfisch-Ikegami model, and 
the bit error rate satisfying the desired QoE value $Q$, respectively. The function $\mathrm{BER}(Q)$ is defined by
\begin{equation}
\mathrm{BER}(Q)  :=  
1-{\rm exp}\left(\frac{1}{S}\ln \left(1+\frac{1}{4.473}\ln\left(\frac{Q-1.065}{3.01}\right)  \right)  \right),
\end{equation}
where $S$ denotes the packet size [bit] in data transmission.

We consider the scenario where 
the SBSs first accommodate the users in their service areas and then
the MBS serves the rest of the users.
Because SBSs are deployed such that they cover disjoint areas due to cost efficiency, 
the sum of the numbers of users  is given by  $\sum_{i=1}^N U_i$, and
the sum of the numbers of users associated with the SBSs is 
$\sum_{i=1}^{N} F(u_i, U_i)$.

\subsection{Optimization problem for cell zooming}
\label{sec:formulation_problem}
The objective of cell zooming is to determine 
the transmission powers and 
the sleep-wake schedules of the SBSs.
The following two elements are evaluated
in the objective function of cell zooming.
\renewcommand{\labelenumi}{\alph{enumi})}
\begin{enumerate}
	\item
	The SBSs use 
	 renewable energy in their operation, whereas the MBS
	 consumes energy from the grid. 
	 To make
	the cell network ``green,'' we should 
	leverage the SBSs. Therefore,
	we minimize
	the number of users  not associated with the SBS,
	$
	U_i[k]  - F(u_i[k],U_i[k]).
	$
	
	\item
	The energy depletion of the SBSs may severely degrade
	the service quality.
	It is preferable that the 
	residual energy of 
	each SBS is close to the full charge state $X_{\rm max}$. 
	In other words, we have to keep
	the difference
	$
	X_{\max} - x_i [k+1]
	$ small, when determining a cell zooming policy at time $k$.
\end{enumerate}

The objective function $P_k$ at time $k$ is given by
a weighted squared sum of the above two terms:
\begin{align}
&P_k\big({\bf u}, {\bf U}[k]\big) :=
\sum_{i=1}^{N}
\bigg[
\Big(
U_i[k] -  F\big(u_i,U_i[k]\big)\Big)^2 \notag \\
&+ \lambda
\Big(
X_{\max} - 
\sat
\big(
x_i[k]+ h(w_i[k]-u_i/\gamma - s_i)
\big)
\Big)^{2}~\!
\bigg],
\label{eq:obj_func}
\end{align}
where 
\begin{equation}
{\bf u} := \begin{bmatrix} u_1 & \cdots & u_N \end{bmatrix},\quad 
{\bf U}[k] := \begin{bmatrix} U_1[k] & \cdots & U_N[k] \end{bmatrix}.
\end{equation}
Here 
a positive quantity $\lambda$ is a weighting parameter.
If this parameter is small, then
the  control policy places more importance on the number of users
associated with the SBSs but less importance on the residual energy.

The optimization problem of cell zooming
can be formulated mathematically as follows:
\begin{align}
&\hspace{-20pt}\min_{\bf u}~ P_k\big({\bf u}, {\bf U}[k]\big) \label{prob:original} \\
\text{s.t.~}&~\text{[C1]~~}	u_i \geq 0,\quad 1 \leq i \leq N, \notag \\
&~\text{[C2]~~}s_i = 
\begin{cases}
s_{\text{active}} & \text{if $u_i> 0$},\\
s_{\text{sleep}} & \text{if $u_i = 0$},
\end{cases} \quad 1 \leq i \leq N, \notag \\
&~\text{[C3]~~}
x_i[k] + h(w_i[k]-u_i/\gamma - s_i) \geq 0,\quad 1 \leq i \leq N, \notag \\
&~\text{[C4]~~}
U_i[k] - F\big(u_i,U_i[k]\big) \geq 0,\quad 1 \leq i \leq N, \notag \\
&~\text{[C5]~~}
\sum_{i=1}^{N} \Big(
U_i[k] - F\big(u_i,U_i[k]\big) \Big) \leq U_{\rm Macro}. \notag
\end{align}

In the above minimization problem, [C1] is
the non-negativity constraint on transmission powers,
and [C2] means that the system power takes only two discrete values depending 
on whether the SBS is in the active mode or the sleep mode.
By [C3], a cell zooming policy is determined to avoid
the energy shortage of the SBSs.
If there exists no transmission power $u_i \geq 0$ satisfying this constraint,
we deactivate the $i$th SBS, i.e., set $u_i = 0$ and $s_i = s_{\text{sleep}}$.
The constraint [C4] means that 
the $i$th SBS determines its transmission power so that 
it accommodates at most $U_i[k]$ users.
If this constraint is violated, then
the SBSs waste their energy because the transmission powers are unnecessarily large.
By [C5], 
the MBS can accommodate
all the users who are not associated with the SBSs.

It is challenging to solve 
the minimization problem~\eqref{prob:original} in real time due to
the constraint [C2].   Indeed,
using an 
auxiliary integer variable $q_i \in \{0,1\}$, we can rewrite the constraint [C2] as
\begin{equation}
q_i u_i = 0,~~-q_i < u_i,~~s_i = (1-q_i)s_{\text{active}}  + q_i s_{\text{sleep}}.
\end{equation}
Therefore,
the minimization problem~\eqref{prob:original} is equivalent to
a mixed-integer nonlinear programming problem. 

\subsection{Importance of data masking}
We consider the situation where 
intruders
can analyze control packets sent by the SBSs and the MBS
but not have any prior knowledge of controllers 
or correlation between data. 
To solve the minimization problem~\eqref{prob:original}, 
the numbers of users, $U_1[k],\dots, U_N[k]$, are transmitted 
from the SBSs to the MBS.
The leakage of such social data leads to security threats and
commercial losses.
To avoid it,
the $i$th SBS sends not the raw data $U_i[k]$ but the masked data $U_i[k] + v_i[k]$, where $v_i[k]$ is the masking noise
for the $i$th SBS at time $k$.
As the masking noise becomes larger, it is more difficult for intruders
to find the exact value of the confidential measurement data.
However, large masking noise degrades the  performance of cell zooming.
The objective of this paper is to develop a cell zooming method that is
robust against masking noise.

\subsection{Comparison between centralized and distributed methods with
	 masking noise}
\label{sec:comparison}
In the case of centralized control, the central controller (often in the MBS) needs to collect the information
on the residual energy $x_i$, the harvested energy $w_i$, and the number of users $U_i$ 
from the $i$th SBS for solving the minimization problem \eqref{prob:original}.
Since $U_i$ is confidential data, each SBS adds the noise $v_i$ to $U_i$ and send
the masked data $U_i + v_i$ to the central controller.
Fig.~\ref{fig:centralized_case} illustrates the information flow of the centralized case.
The resulting minimization problem that the central controller solves is given by
\begin{align}
&\hspace{-17pt}\min_{\bf u}~ P_k\big({\bf u}, {\bf U}[k]+{\bf v}[k]\big) \label{prob:CC_noise} \\
\text{s.t.~}& \text{[C1]\hspace{1pt}--\hspace{1pt}[C4]}, \notag\\
&\text{[C5 with noise]~~}\notag \\
&\sum_{i=1}^{N} \Big(
U_i[k]+v_i[k] - F\big(u_i,U_i[k]+v_i[k]\big) \Big) \leq U_{\rm Macro}, \notag
\end{align}
where 
\[
{\bf v}[k] := \begin{bmatrix} v_1[k] & \cdots & v_N[k] \end{bmatrix}.
\]
Notice that [C4] is not affected by the masking noise because
\[
U+v - F(u,U+v) \geq 0
\] 
if and only if $u \leq r^{-\frac{19}{10}}$
for every $v > -U$.
In the above minimization problem, both  the objective function
and the constraint [C5] contain the noise. Therefore, 
the  solution of the minimization problem \eqref{prob:CC_noise} may have
a large error due to the masking noise.

\begin{figure}
	\centering
	\subcaptionbox{Centralized case.	\vspace{10pt} 
		\label{fig:centralized_case}}
	{\includegraphics[width = 6.4cm,clip]{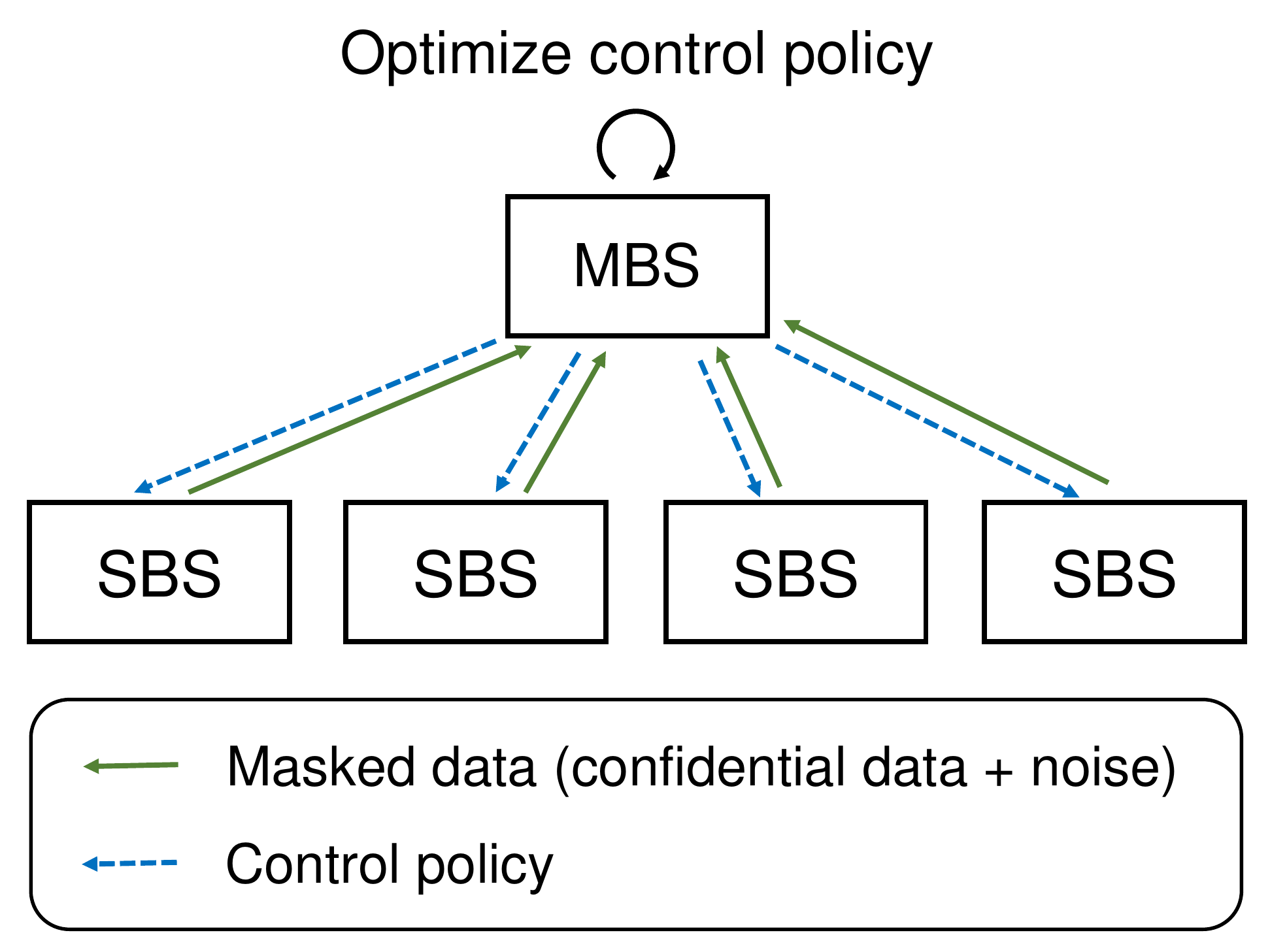}}
	\subcaptionbox{Distributed case.
		\label{fig:distributed_case}}
	{\includegraphics[width = 6.4cm,clip]{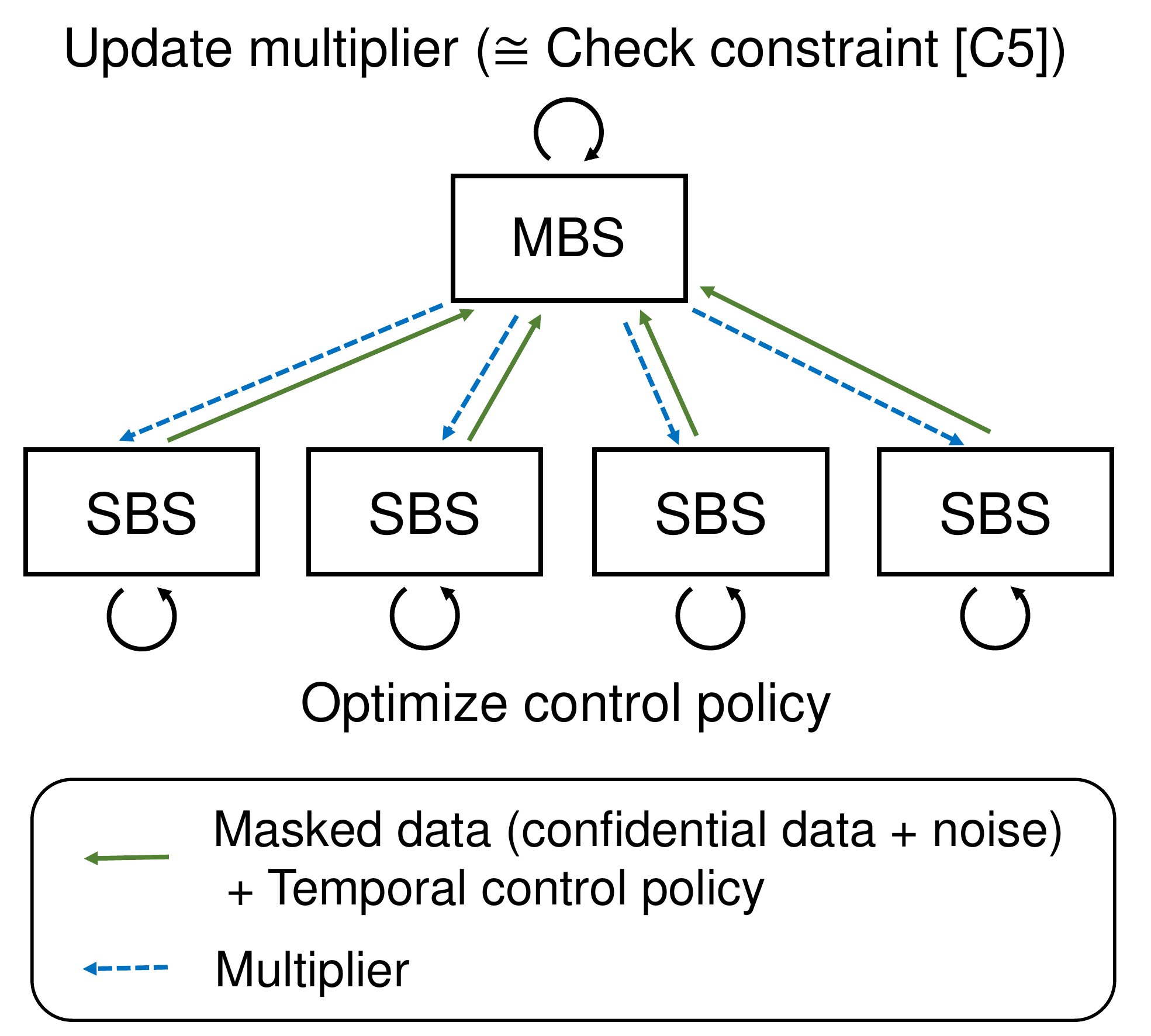}}
	\caption{Information flow. \label{fig:AC}}
\end{figure}

The proposed method employs distributed optimization.
In the distributed approach, the objective function $P_k$ is divided into $N$ local functions
and is minimized locally in each SBS.
Consequently, 
the masking noise affects only the constraint [C5], which
makes the proposed method resistant to masking noise.

As shown in
Fig.~\ref{fig:distributed_case},
each SBS locally solves a divided minimization problem, and
the central controller simply checks
whether the global constraint [C5] is satisfied for the masked data $U_i + v_i$ and the temporal
control policy transmitted from each SBS.
Updating a multiplier in Fig.~\ref{fig:distributed_case} represents this action of the central controller,
because we apply the duality principle (see, e.g., \cite[Chapter~5]{Bertsekas1999}) to the minimization problem.
From the updated multiplier,
the SBSs, which know only local data, 
learn the extent to which the global constraint [C5] is satisfied or violated, and
then update the control policy based on this learning.
A detailed discussion on it is provided in Section~\ref{sec:Algorithm_DCZ}.

\begin{remark}
	Wada and Sakurama \cite{Wada2017} propose a masking method to
	protect the agent privacy for distributed optimization, where an ``agent'' is
	an SBS in our cell zooming problem.
	By this method, each agent conceals private information from other agents
	without any influence on solutions of optimization problems.
	Its essential idea is to exchange 
	masking signals between neighbor agents. However,
	it does not work in our setting where
	intruders analyze control packets transmitted from the SBSs.
\end{remark}

\section{Distributed Cell Zooming}
\label{sec:Distributed_cell_zooming}
In this section, we propose a distributed cell zooming algorithm.
First, we approximate the minimization problem \eqref{prob:original} 
formulated in the previous section.
Using the approximate problem, we next apply a distributed optimization technique to
cell zooming for off-grid SBSs.
Finally, we analyze the divided minimization problem that each SBS
iteratively solves in the distributed algorithm.

\subsection{Approximation techniques}
\label{sec:approximation_tech}
It is important in distributed cell zooming to
reduce computational complexity. The reason is that
a divided minimization problem is iteratively 
solved by each SBS, which
has limited computational resources in general.
Computing the exact solution of
the originally formulated minimization problem \eqref{prob:original} 
requires significant computational resources. 
Hence we first develop approximation techniques 
to reduce the computational complexity.

\subsubsection{Omission of saturation function}
The saturation function $\sat$ in 
the second term of the objective function $P_k$
has little effect on solutions, and therefore
we omit it. In fact,
by the constraint [C3], the input of $\sat$,
\begin{equation}
\label{eq:residual_energy}
x_i[k] + h(w_i[k]-u_i/\gamma - s_i),
\end{equation} 
is non-negative, and 
hence we can omit the lower bound of $\sat$. 
Moreover, 
the omission of the upper bound of $\sat$
yields only a small approximation error. This is because
if the residual energy is close to $X_{\max}$, then 
the corresponding transmission power is set to a large value so that 
the input of $\sat$ given in \eqref{eq:residual_energy} becomes smaller than 
the upper bound $X_{\max}$.
The omission of $\sat$ 
makes the minimization problem easy to deal with, because
this function is not convex or concave; see the definition \eqref{eq:sat_def}.

\subsubsection{Convex relaxation by $\ell^1$ norm}
\label{sec:convex_relaxation}
By the omission of the saturation function,
the second term  of $P_k$ can be rewritten as
\begin{align}
&\big(
X_{\max} - 
x_i[k] - h(w_i[k]-u_i/\gamma - s_i)
\big)^{2} \notag \\
&\quad =
(X_{\max} - 
x_i[k] - hw_i[k] + h u_i/\gamma)^2 +h^2 s_i^2 \label{eq:sec_term} \\
&\qquad + 
2h(X_{\max} - 
x_i[k] - hw_i[k] + h u_i/\gamma)s_i. \notag 
\end{align}
Using the $\ell^0$ norm $\|{\bf u}\|_0$, we obtain
\begin{equation}
\label{eq:l0norm}
\sum_{i=1}^N
h^2 s_i^2 = 
h^2 \big(s_{\rm active}^2 - s_{\rm sleep}^2\big) \|{\bf u}\|_0 + 
h^2 Ns_{\rm sleep}^2.
\end{equation}
We	
approximate this non-convex term 
by the $\ell^1$ norm $\|{\bf u}\|_1$:
\begin{align}
&h^2 \big(s_{\rm active}^2 - s_{\rm sleep}^2\big) \|{\bf u}\|_0 + 
h^2 Ns_{\rm sleep}^2 \notag \\
&\quad\approx
h^2 \big(s_{\rm active}^2 - s_{\rm sleep}^2\big) \|{\bf u}\|_1 + 
h^2 Ns_{\rm sleep}^2. \label{eq:l0norm_app} 
\end{align}
This
relaxation
is commonly used in the field of signal processing  \cite{Donoho2006} and 
has been recently applied to distributed optimization in \cite{HayashiSICEAC2018}.
\subsubsection{Use of previous system power}
\label{sec:use_previous}
The right-hand side of \eqref{eq:sec_term} has the non-convex term 
\[
2h(X_{\max} - 
x_i[k] - hw_i[k] + h u_i/\gamma)s_i.
\] 
We replace $s_i$ by the previous system power $s_i[k-1]$, i.e., 
\begin{align}
&2h(X_{\max} - 
x_i[k] - hw_i[k] + h u_i/\gamma)s_i \notag \\
&\qquad \approx
2h(X_{\max} - 
x_i[k] - hw_i[k] + h u_i/\gamma)s_i[k-1].
\label{eq:previous_s}
\end{align}

\subsubsection{Margin for constraint {\rm [C5]}}
It is known that decreasing the $\ell^1$ norm $\|{\bf u}\|_1$ makes 
the vector $\bf u$ sparse (see, e.g., \cite[Chapters~2\hspace{1pt}--\hspace{1pt}4]{Nagahara2020}), but
some elements may not be equal to zero exactly. Hence
we will use the non-negative quantity $u_{\rm thres}$ as a threshold in the proposed algorithm.
If the solution $u_i$ is smaller than the threshold, then
$u_i$ is reset to $0$.
This enables us to avoid ineffective energy usage due to activation with too small transmission power.
However, the reset of transmission powers decreases the number of 
users associated with the SBSs. To guarantee that all users are accommodated, i.e., the constraint [C5] is satisfied,
we replace [C5] by a slightly stricter constraint
\begin{align}
\text{[C5']~~}\sum_{i=1}^{N}\Big(
U_i[k] - F\big(u_i,U_i[k]\big) \Big) \leq (1-c)U_{\rm Macro},
\label{eq:margin_for_FA}
\end{align}
where a non-negative quantity $c$ is a margin constant for full accommodation.

\subsubsection{Approximate problem}
Employing the 
techniques~1)\hspace{1pt}--\hspace{1pt}4) above, we approximate
the original problem \eqref{prob:original} by
a minimization problem
of the form 
\begin{align}
\label{eq:app_min_pro2}
\min_{\substack{0 \leq u_i \leq u^{\max}_i[k] \\ i=1,\dots,N}}~\sum_{i=1}^N f_{i,k}\big(u_i,U_i[k]\big)
\text{~~s.t.~~}
\sum_{i=1}^N  g\big(u_i,U_i[k]\big) \leq 0,
\end{align}
where the non-negative quantity $u^{\max}_i[k]$ and the functions $f_{i,k}$, $g$
are defined by \eqref{eq:fgik_def} at the top on the next page.
The detailed derivation of \eqref{eq:app_min_pro2} can be found in Appendix~\ref{sec:appA}.

In \eqref{eq:app_min_pro2},
the constraints
[C1], [C3], and [C4] of the original problem \eqref{prob:original} are transformed into
$0 \leq u_i \leq u^{\max}_i[k]$, and [C5'] is
$\sum_{i=1}^N  g(u_i,U_i[k]) \leq 0$.
The constraint [C2] in \eqref{prob:original}
is omitted by the approximation techniques 2) and 3).

\newcounter{mytempeqncnt2}
\begin{figure*}
	\begin{subequations}
		\label{eq:fgik_def}
		\noeqref{eq:umax,eq:fik,eq:gik}
		\begin{align}
		u^{\max}_i[k] &:=
		\max
		\Big\{0,~	\min\big\{
		r^{-\frac{19}{10}},~	
		\gamma (x_i[k]/h+ w_i[k] - s_{\rm active}) 	\big\}
		\Big\} 
		\label{eq:umax}\\
		f_{i,k}(u,U) &:=
		\big(
		U -  F(u,U)\big)^2 
		+\lambda
		(
		X_{\max} - 
		x_i[k] - hw_i[k]+hu/\gamma)^2 +
		\lambda h^2
		\big(
		s_{\rm active}^2 - s_{\rm sleep}^2 + 2 s_i[k-1]/\gamma
		\big)
		u \label{eq:fik}\\
		g(u,U) &:= U - F(u,U) -  \frac{(1-c)U_{\text{Macro}}}{N }
		\label{eq:gik} 
		\end{align}
	\end{subequations}
	\hrulefill
	\vspace{-10pt}
\end{figure*}

\subsection{Distributed algorithm for cell zooming}
\label{sec:Algorithm_DCZ}

\begin{algorithm}[!tb]
	\caption{Distributed cell zooming of $i$th SBS at time $k$}
	\label{algo:DCZ}
	\begin{algorithmic}[1]
		\Require $x_i[k]$, $w_i[k]$, $U_i[k]$, $s_i[k-1]$, $\mu = \mu[k-1]$
		\Ensure  $u_i[k]$, $s_i[k]$, $\mu[k]$
		\State Initialize $t=1$
		\State Generate the masking noise $v_i[k]$
		\State Transmit $U_i[k] + v_i[k]$  to MBS
		\State {\bf while} $t \leq T $ {\bf do} 
		\State \quad Compute the solution $u_i^*$ of \eqref{eq:local_minimization_prob}
		~(see Theorem~\ref{thm:third_order})
		\State \quad
		Transmit $u_i^*$  to MBS
		\State \quad
		Receive a new multiplier from MBS and update: \\
		\qquad
		$\mu \leftarrow 	\max\big\{0,~
		\mu +
		\alpha^{[t]} \sum_{i=1}^N g\big(u_i^*,U_i[k]+v_i[k]\big)\big\}
		$
		\State \quad
		$t  \leftarrow t+1$ 
		\State {\bf end while} 
		\State {\bf if} $u_i^* \leq u_{\rm thres}$ {\bf then} 
		\State \quad $u_i[k] \leftarrow 0$, $s_i[k] \leftarrow  s_{\rm sleep}$, $\mu[k] \leftarrow \mu$
		\State {\bf else} 
		\State \quad $u_i[k] \leftarrow u_i^*$, $s_i[k] \leftarrow s_{\rm active}$, $\mu[k] \leftarrow \mu$
		\State {\bf end if} 
	\end{algorithmic}
\end{algorithm}

Based on the approximate problem \eqref{eq:app_min_pro2} 
introduced in Section~\ref{sec:approximation_tech},
we propose Algorithm~\ref{algo:DCZ} for distributed cell zooming with masked data,
where $T$ and $\alpha^{[t]}$ denote the terminal step and
the stepsize for iteration $t$, respectively. The terminal step $T$ 
represents how many communications between the SBSs and the MBS are required 
to determine the cell zooming policy at each time $k$.

Algorithm~\ref{algo:DCZ} consists of two parts. 
In lines 4\hspace{1pt}--\hspace{1pt}10, the SBSs and the MBS cooperatively 
solve the minimization problem
\begin{align}
&\min_{\substack{0 \leq u_i \leq u^{\max}_i[k] \\ i=1,\dots,N}}~\sum_{i=1}^N f_{i,k}\big(u_i,U_i[k]\big)
\label{eq:app_min_pro2_noise} \\
&\text{~~s.t.~~}
\sum_{i=1}^N  g\big(u_i,U_i[k]+v_i[k]\big) \leq 0. \notag
\end{align}
The constraint $\sum_{i=1}^N  g(u_i,U_i[k]+v_i[k]) \leq 0$  is equivalent to
[C5 with noise] in \eqref{prob:CC_noise} if the margin constant $c$ is zero. 
In lines 11\hspace{1pt}--\hspace{1pt}15, the SBS resets the transmission power $u_i$ to zero
if $u_i$ is smaller than the threshold $u_{\rm thres}$, which
prevents the SBSs from operating inefficiently for the accommodation of
only a few users.

One can easily see that $f_{i,k}(u,U)$ and $g(u,U)$ 
are 
strictly convex with respect to the first variable $u$.
Therefore,
$u_1^*,\dots,u_N^*$ in lines 4\hspace{1pt}--\hspace{1pt}10 
of Algorithm~\ref{algo:DCZ} converge to the 
unique solution of the  problem \eqref{eq:app_min_pro2_noise}
as $T \to \infty$,
if the following two conditions hold:
\begin{enumerate}
	\item There exist $0 \leq  u_i \leq u^{\max}_i[k]$, $i=1,\dots,N$, such that
	\[
	\sum_{i=1}^N  g\big(u_i,U_i[k] + v_i[k]\big)  < 0
	\] 
	is satisfied (Slater's condition);
	\item The stepsize $\alpha^{[t]}$ satisfies
			\begin{align}
		\label{eq:alpha_beta_cond}
		\sum_{t=1}^{\infty} \alpha^{[t]} = +\infty,\quad 
		\sum_{t=1}^{\infty} \Big(\alpha^{[t]}\Big)^2 < +\infty.
		\end{align}
\end{enumerate}
For example, $\alpha^{[t]} = 1/t$
satisfies the conditions 
\eqref{eq:alpha_beta_cond}.
See, e.g., \cite[Chapter~6]{Bertsekas1999} for this convergence result.

In Algorithm~\ref{algo:DCZ},
the minimization problem \eqref{eq:app_min_pro2_noise} is solved
based on the  duality principle.
In other words, the SBSs and the MBS compute the solution of the primal problem \eqref{eq:app_min_pro2_noise},
by solving
the dual problem of \eqref{eq:app_min_pro2_noise},
$\min_{\mu \geq 0} H(\mu)$,
where the dual function $H$ is given by
\begin{align}
&H(\mu) := \\
&\min_{\substack{0 \leq u_i \leq u^{\max}_i[k] \\ i=1,\dots,N}} \sum_{i=1}^N f_{i,k}\big(u_i,U_i[k]\big)+\mu g\big(u_i,U_i [k]+v_i[k]\big). \notag
\end{align}
In what follows, we summarize how to solve this 
dual problem by the subgradient method; see \cite[Chapters~5,
6]{Bertsekas1999} and 
\cite{Yang2010} for details.

The first step is to obtain
a subgradient of the dual function $H$ for a fixed multiplier $\mu$.
It is known that 
\begin{equation}
\label{eq:subgradient}
\sum_{i=1}^N g\big(u_i^*,U_i [k]+v_i[k]\big)
\end{equation}
is a subgradient of  the dual function $H$ at $\mu$, where
$u_i^*$ is a temporal transmission power defined by
\begin{align}
&\begin{bmatrix}
u_1^* & \cdots & u_N^*
\end{bmatrix}
:= \label{eq:hat_u}\\
&\argmin_{\substack{0 \leq u_i \leq u^{\max}_i[k] \\ i=1,\dots,N}} \sum_{i=1}^N f_{i,k}\big(u_i,U_i[k]\big)+\mu g\big(u_i,U_i [k]+v_i[k]\big). \notag
\end{align}
The global problem in the right-hand side of \eqref{eq:hat_u} 
has a separable structure and
can be decomposed
into $N$ local problems:
\begin{align}
&\min_{\substack{0 \leq u_i \leq u^{\max}_i[k] \\ i=1,\dots,N}} \sum_{i=1}^N f_{i,k}\big(u_i,U_i[k]\big)+\mu g\big(u_i,U_i [k]+v_i[k]\big) \\
&~= \sum_{i=1}^N
\min_{0 \leq u_i \leq u^{\max}_i[k]}
f_{i,k}\big(u_i,U_i[k]\big)+\mu g\big(u_i,U_i [k]+v_i[k]\big). \notag
\end{align}
This observation plays an important role in distributed optimization.
In fact, $u_i^*$ is given by a solution of the local problem
\begin{equation}
\label{eq:local_minimization_prob}
\min_{0 \leq u_i \leq u^{\max}_i[k]}
f_{i,k}\big(u_i,U_i[k]\big)+\mu g\big(u_i,U_i [k]+v_i[k]\big)
\end{equation}
for every $i=1,\dots,N$.
This means that 
the temporal transmission power $u_i^*$ can be computed
locally in each SBS.
Collecting the temporal transmission powers $u_i^*$ and
the masked data $U_i[k] + v_i[k]$ from the SBSs,
the MBS computes the subgradient given by \eqref{eq:subgradient}.

The second step is to update the multiplier $\mu$:
\begin{equation}
\label{eq:multiplier_update}
\mu~~ \leftarrow~~
\max\left\{0,~
\mu +
\alpha^{[t]} \sum_{i=1}^N g_{i,k}\big(u_i^*,U_i[k]+v_i[k]\big)\right\}.
\end{equation}
Roughly speaking, updating the multiplier is
a mechanism to check whether 
the global constraint 
\[
\sum_{i=1}^N g_{i,k}\big(u_i^*,U_i[k]+v_i[k]\big) \leq 0
\]
is satisfied.
In fact, if this constraint is not satisfied, then the update rule \eqref{eq:multiplier_update}
increases the multiplier $\mu$. Moreover, for a sufficiently large multiplier $\mu$,
a solution of the minimization problem \eqref{eq:hat_u}
must satisfy the constraint.

The MBS transfers a new multiplier to the SBSs, 
and the SBSs learn from it the extent to which
the global constraint is satisfied or violated. Based on this learning,
the SBSs compute temporal transmission power again.
Repeating these two steps yields
the solution of the minimization problem \eqref{eq:app_min_pro2_noise}.

\begin{remark}[Complexity of Algorithm \ref{algo:DCZ}]
	\label{rem:complexity}
	At every time $k$,
	the $i$th SBS solves $T$ times  the minimization problem \eqref{eq:local_minimization_prob} and
	communicates $2T+1$ times with the MBS.
	As we will show in the simulation section,
	$T = 20$ is enough for the outputs of Algorithm \ref{algo:DCZ} to converge. 
	Since the data transfer time is at most a few millisecond, 
	the bottleneck of Algorithm \ref{algo:DCZ} is to solve 
	the minimization problem \eqref{eq:local_minimization_prob}.
	In the next subsection, however, we provide an explicit formula for
	an approximate solution of  the minimization problem \eqref{eq:local_minimization_prob},
	which completely resolves the computational issue.
	The resulting computational complexity of Algorithm~\ref{algo:DCZ}
	is $O(T)$ for each SBS, and therefore the total complexity of the cell network
	is $O(NT)$. 
	It is worth mentioning that Algorithm~\ref{algo:DCZ}
	is dimension-free in the sense that 
	the computational complexity of each SBS does not depend on $N$.
	On the other hand, if we solve the minimization problem~\eqref{prob:original}
	without any approximation, then
	the worst-case complexity is at least $O(2^N)$ because
	we have to solve $2^N$ optimization problems with fixed 
	sleep-wake schedules.
\end{remark}

\subsection{Solution of minimization problem (15)}
We characterize the solution of the minimization problem \eqref{eq:local_minimization_prob} by a simple nonlinear equation and
provide an explicit formula for its approximate solution, inspired by Theorem~3 of \cite{Wakaiki2018EHSCN}.
The formula obtained from Theorem~\ref{thm:third_order} below 
significantly reduces the computational complexity of Algorithm~\ref{algo:DCZ}.
The proof can be found in Appendix~\ref{sec:appB}.
For simplicity of notation, we omit the subscripts $i,k$ and $[k]$ in the following theorem.
\begin{theorem}
	\label{thm:third_order}
	\textit{Define the coefficients
		\begin{align*}
		p_1 &:=
		\frac{2\lambda h^2}{\gamma^2},\quad 
		p_2 := 
		\frac{20 r^2 U^2 }{19} \\
		p_3 &:=
		\frac{2\lambda h}{\gamma}(X_{\max} - x - hw ) \notag \\
		&\hspace{30pt}+ \lambda h^2
		\left(s_{\rm active}^2 - s_{\rm sleep}^2 + 
		\frac{2s[k-1]}{\gamma}
		\right)  \\
		p_4 &:=
		\frac{10r}{19} \big(2U^2+\mu(U+v)\big), 
		\end{align*}
		and suppose that $p_4 >0$.
	}
	
	a) 
	\textit{
		The solution $u^*$ of the minimization problem \eqref{eq:local_minimization_prob} is given by
		\begin{equation}
		\label{eq:u_opt}
		u^* = 
		\begin{cases}
		u^{\sol} & \text{if $0 \leq u^{\sol} \leq u^{\max}$}, \\
		u^{\max} & \text{if $u^{\sol} > u^{\max}$},
		\end{cases}
		\end{equation}
		where $u^{\sol}$ is a unique positive solution $u^{\sol}$ of the nonlinear equation
		\begin{equation}
		p_1 u^{\frac{28}{19}} + 	p_2 u^{\frac{10}{19}} + p_3 u^{\frac{9}{19}} - p_4 = 0.
		\end{equation}
	}
	
	b)
	\textit{Moreover, if $p_3 \geq 0$, then $u^{\sol}$ satisfies
		\begin{equation}
		\label{eq:v*_app}
		\begin{cases}
		\chi^{\frac{19}{9}} \leq u^{\sol} \leq \chi^{\frac{19}{10}} & \text{if $\chi \leq 1 $}, \\
		\chi^{\frac{19}{10}} \leq u^{\sol} \leq \chi^{\frac{19}{9}} & \text{if $\chi > 1$},
		\end{cases}
		\end{equation}
		where 
		\begin{align}
		\label{eq:third_order_eq_solution}
		\chi := \sqrt[3]{\xi + \eta} + \sqrt[3]{\xi - \eta}
		\end{align}
		with
		\begin{align}
		\xi := \frac{p_4}{2p_1},\quad
		\eta := \sqrt{
			\frac{p_4^2}{4p_1^2} + \frac{(p_2+p_3)^3}{27p_1^3}
		} .
		\end{align}
	}
\end{theorem}
\vspace{5pt}

Using Theorem~\ref{thm:third_order}, 
we obtain an explicit formula for an 
approximate solution $u^{\text{app}}$  of the minimization problem \eqref{eq:local_minimization_prob}:
\begin{equation}
\label{eq:app_solution}
u^{\text{app}} := 
\begin{cases}
\chi^2 & \text{if $0 \leq \chi^2 \leq u^{\max}$}, \\
u^{\max} & \text{if $\chi^2 > u^{\max}$}.
\end{cases}
\end{equation}
It is worthy to note that 
$\chi^2$ can be obtained only from
the four basic arithmetic operations and the computation of 
square and cubic roots.
Hence the computational cost of
the approximation solution $u^{\text{app}}$ is quite small.

\section{Differential privacy in cell zooming}
\label{sec:differential_privacy}
In line 3 of Algorithm~\ref{algo:DCZ},
the SBS 
transmits  the masked data $U_i[k] + v_i[k]$ to the MBS.
As the noise amplitude $|v_i[k]|$ becomes larger, it is more difficult for intruders to
estimate the exact number of users, $U_i[k]$, from the masked data $U_i[k] + v_i[k]$. However,
large noise degrades control performance.
In this section, we investigate the relationship between 
the performance of cell zooming and 
the noise intensity
from the viewpoint of differential privacy.
Differential privacy
has been recently applied to various areas such as 
interesting location pattern mining~\cite{Ho2013} and
power usage data analysis in a smart grid~\cite{Liao2019};
see also the survey \cite{Hassan2020}.

\subsection{Notion of differential privacy}
To define differential privacy,
we first establish an adjacency relation on data sets.
In this paper, we use the notion of $\delta$-adjacency
introduced in \cite{Nozari2017, Weerakkody2019}.
\begin{definition}
	We say that $({\bf U},\widetilde {\bf U}) \in \mathbb{R}^N \times \mathbb{R}^N$ is {\em $\delta$-adjacent} if 
	${\bf U}$ and $\widetilde {\bf U}$ have at most one different element
	and the difference does not exceed 
	$\delta$, that is, $\|{\bf U} - \widetilde {\bf U}\|_0 \leq 1$ and $\|{\bf U} - \widetilde {\bf U}\|_1 \leq \delta$.
\end{definition}

In the definition of $\delta$-adjacency, the parameter
$\delta$  plays a role similar to the $L^1$ sensitivity of a query in the context 
of the standard differential privacy \cite{Dwork2014}.
\begin{definition}
	For a vector-valued random variable ${\bf v}$, the mechanism ${\bf \Theta} ({\bf U}, {\bf v}) = {\bf U} + {\bf v}$ is said to be
	{\em $\epsilon$-differentially private for 
		$\delta$-adjacent pairs} if 
	\begin{equation}
	\label{eq:diff_privacy}
	\textrm{P}	({\bf \Theta} ({\bf U}, {\bf v}) \in S) \leq 
	e^{\epsilon}  \textrm{P}	 ({\bf \Theta} (\widetilde {\bf U}, {\bf v})\in S)
	\end{equation}
	holds for every $\delta$-adjacent pair $(\bf U, \widetilde {\bf U})$ and
	every (Borel measurable) set $S$.
\end{definition}

We can rewrite \eqref{eq:diff_privacy} as
\begin{equation}
\big|
\ln \textrm{P}	({\bf \Theta} ({\bf U}, {\bf v}) \in S) - 
\ln  \textrm{P}	 ({\bf \Theta} (\widetilde {\bf U}, {\bf v}) \in S)
\big| \leq \epsilon,
\end{equation}
which says that for every $\delta$-adjacent pair  $({\bf U}, \widetilde {\bf U})$,
the distributions over the masked data should be close.
In other words, for sufficiently small $\epsilon$,
intruders cannot distinguish $U_i$ and $U_i+\delta_i$ with $|\delta_i| \leq \delta$
from the masked data ${\bf U} + {\bf v}$
even if they
know all the other numbers of users, $U_j$ with $j \not= i$.
The parameters $\delta$ and $\epsilon$ are determined by the security policy.
For large $\delta$, a wide range of numbers of users is secret.
As $\epsilon$ decreases, the estimation of the number of users becomes more difficult.
For example, $0.1 \leq \epsilon \leq 2$ is used in the study \cite{Ho2013} 
on preserving privacy for interesting location discovery from
history records of individual locations.

\subsection{Constraint error due to masking noise}
\label{sec:constraint_error}
We next investigate the effect of the masking noise to the minimization problem \eqref{eq:app_min_pro2_noise}. Recall that
only the constraint function $g$ is affected by the masking noise.
Since
\begin{equation}
g(u,U+v) - g(u,U) = v\Big(1-ru^{\frac{10}{19}}\Big),
\end{equation}
the constraint  changes as follows:
\begin{align}
&\text{[Without noise]} \quad \sum_{i=1}^N g_i(u_i,U_i) \leq 0 \\
&~~\to ~~
\text{[With noise]} \quad \sum_{i=1}^N g_i(u_i,U_i) \leq \sum_{i=1}^N v_i\Big(1 - r u_i^{\frac{10}{19}}\Big). \notag 
\end{align}
The error clearly becomes larger as we increase the noise intensity
so as to achieve a high level of confidentiality.
Note that 
\[
0\leq 1 - r u_i^{\frac{10}{19}} \leq 1
\]
by the definition of $u_i^{\max}$ given in \eqref{eq:umax}.
As $1 - r u_i^{\frac{10}{19}}$ decreases, 
the effect of the noise
\[
\Big|v_i\Big(1 - r u_i^{\frac{10}{19}}\Big)\Big|
\]
becomes smaller.
In other words,
when the transmission powers are sufficiently large,
the masking noise has almost no effect on the constraint.
In the next subsection,
we study the sum of the noise signals, $\sum_{i=1}^N v_i$, 
for the worst-case scenario $u_1 =\cdots =u_N = 0$.

\subsection{Trade-off between confidentiality and  accuracy}
We denote by $\Lap(\rho)$ the Laplace distribution with mean zero
and positive scale parameter $\rho$.
When a random variable $v$ is distributed according to $\Lap(\rho)$,
we write $v \sim \Lap(\rho)$, and its probability density function is given by
\[
f(y) = \frac{e^{-\frac{|y|}{\rho}}}{2\rho}.
\]
The parameter $\rho$ represents the intensity of the noise $v \sim \Lap(\rho)$.
For $\bf v = \begin{bmatrix} v_1 & \cdots & v_N
\end{bmatrix}$ with $v_i$ independent and identically distributed and $v_i \sim \Lap(\rho)$,
we write ${\bf v} \sim \Lap(\rho)^N$.

Based on the well-known Laplace mechanism 
(see, e.g.,
\cite[Theorem~3.6]{Dwork2014} and 
\cite[Section~5.1]{Weerakkody2019}), 
we relate
 confidentiality to optimization accuracy. 
Note that the noise intensity $\rho$ is given by $\delta/\epsilon$
in the proposition below.
The proof can be found in Appendix~\ref{sec:appC}.
\begin{proposition}
	\label{prop:diff_privacy}
	\textit{The mechanism ${\bf \Theta} ({\bf U}, {\bf v}) = {\bf U} + {\bf v}$ with 
		${\bf v} \sim \Lap(\delta/\epsilon)^N$  is 
		$\epsilon$-differentially private  for 
		$\delta$-adjacent pairs
		and further $\left| \sum_{i=1}^N v_i \right|$ 
		exceeds the threshold $\Lambda$ with probability less than $\zeta$, that is, 
		\begin{equation}
		\label{eq:sum_noise}
		\textrm{\em P}	\left(
		\left| \sum_{i=1}^N v_i \right| > \Lambda
		\right)  < \zeta,
		\end{equation}
		if the decreasing  function $\psi_N$ on $(0,\infty)$ defined by 
		\begin{equation}
		\label{eq:psi_definition}
		\psi_N(y) := \frac{y^{2N}e^{N-\sqrt{y^2+N^2}}}{ (2N)^N \big(\sqrt{y^2+N^2} - N\big)^N}
		\end{equation}
		satisfies
		\begin{equation}
		\label{eq:differential_privacy_error}
		2\psi_N(\epsilon \Lambda /\delta) < \zeta.
		\end{equation}
	}
\end{proposition}
\vspace{5pt}

Fig.~\ref{fig:psi_plot} shows the function $\psi_N$ defined by \eqref{eq:psi_definition}.
As mentioned in Proposition~\ref{prop:diff_privacy}, $\psi_N(y)$ is decreasing with respect to $y$.
Moreover, we observe from Fig.~\ref{fig:psi_plot} that $\psi_N(y)$ becomes larger as $N$ increases.
\begin{figure}[tb]
	\centering
	\includegraphics[width = 7.3cm]{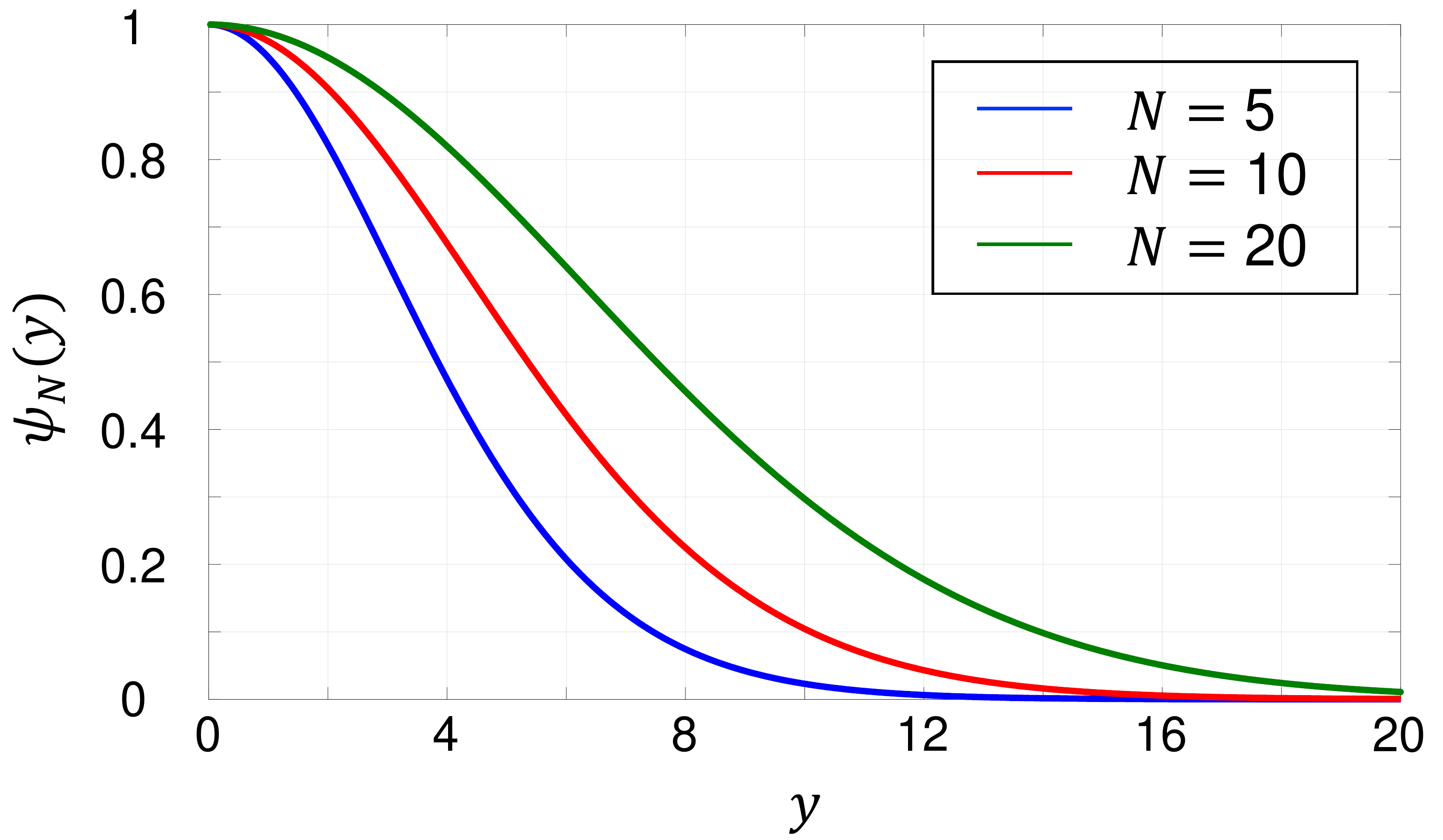}
	\caption{Plot of $\psi_N$.}
	\label{fig:psi_plot}
\end{figure}

The use of Berstein's inequality (see, e.g., \cite[Theorem~1.13]{Rigollet2017}) 
or Theorem~6.5 of \cite{Gupta2012} for the tail bound \eqref{eq:sum_noise} 
yields 
a simple but slightly conservative condition. In fact,
if we use Berstein's inequality, then \eqref{eq:differential_privacy_error} is replaced by 
the following simple condition:
\begin{equation}
\label{eq:differential_privacy_error_BI}
\frac{\delta}{\epsilon} < 
\begin{cases}
\dfrac{\Lambda}{4 \ln(2/\zeta)} & \text{if $N \leq 2\ln(2/\zeta)$}, \vspace{10pt}\\
\dfrac{\Lambda}{2\sqrt{2N \ln(2/\zeta)}} & \text{if $N > 2\ln(2/\zeta)$}.
\end{cases}
\end{equation}

In Fig.~\ref{fig:delta_eps_plot}, the circles and squares indicate the maximum value of $\delta/\epsilon$
satisfying \eqref{eq:differential_privacy_error} and \eqref{eq:differential_privacy_error_BI} with 
$\zeta = 0.01$ and $\Lambda = 30$, respectively.
As $\delta/\epsilon$ becomes large, we can choose smaller $\epsilon$ and larger $\delta$. This implies that 
it is more difficult (as $\epsilon$ becomes smaller) 
to distinguish the correct data and a wider range 
of data (as $\delta$ becomes larger).
Note that \eqref{eq:differential_privacy_error} and \eqref{eq:differential_privacy_error_BI}
are obtained as sufficient conditions for  	${\bf v} \sim \Lap(\delta/\epsilon)^N$ 
to satisfy \eqref{eq:sum_noise}.
The diamonds in Fig.~\ref{fig:delta_eps_plot} 
are 
the maximum value of $\delta/\epsilon$ for which $\zeta$ is larger than
the ratio of
the number of samples ${\bf v} \sim \Lap(\delta/\epsilon)^N$ satisfying $\left| \sum_{i=1}^N v_i \right| > \Lambda$
to the total number of samples, $2 \times 10^5$.
We regard the diamonds as the exact value obtained numerically.
Compared with this sampling approach,
the computational cost of $\delta/\epsilon$ by
the proposed condition \eqref{eq:differential_privacy_error}
is very low. Moreover,
we observe from Fig.~\ref{fig:delta_eps_plot} that 
the proposed condition \eqref{eq:differential_privacy_error} is less conservative than
the conventional condition \eqref{eq:differential_privacy_error_BI}, in particular, for small $N$.
\begin{figure}[tb]
	\centering
	\includegraphics[width = 7.5cm]{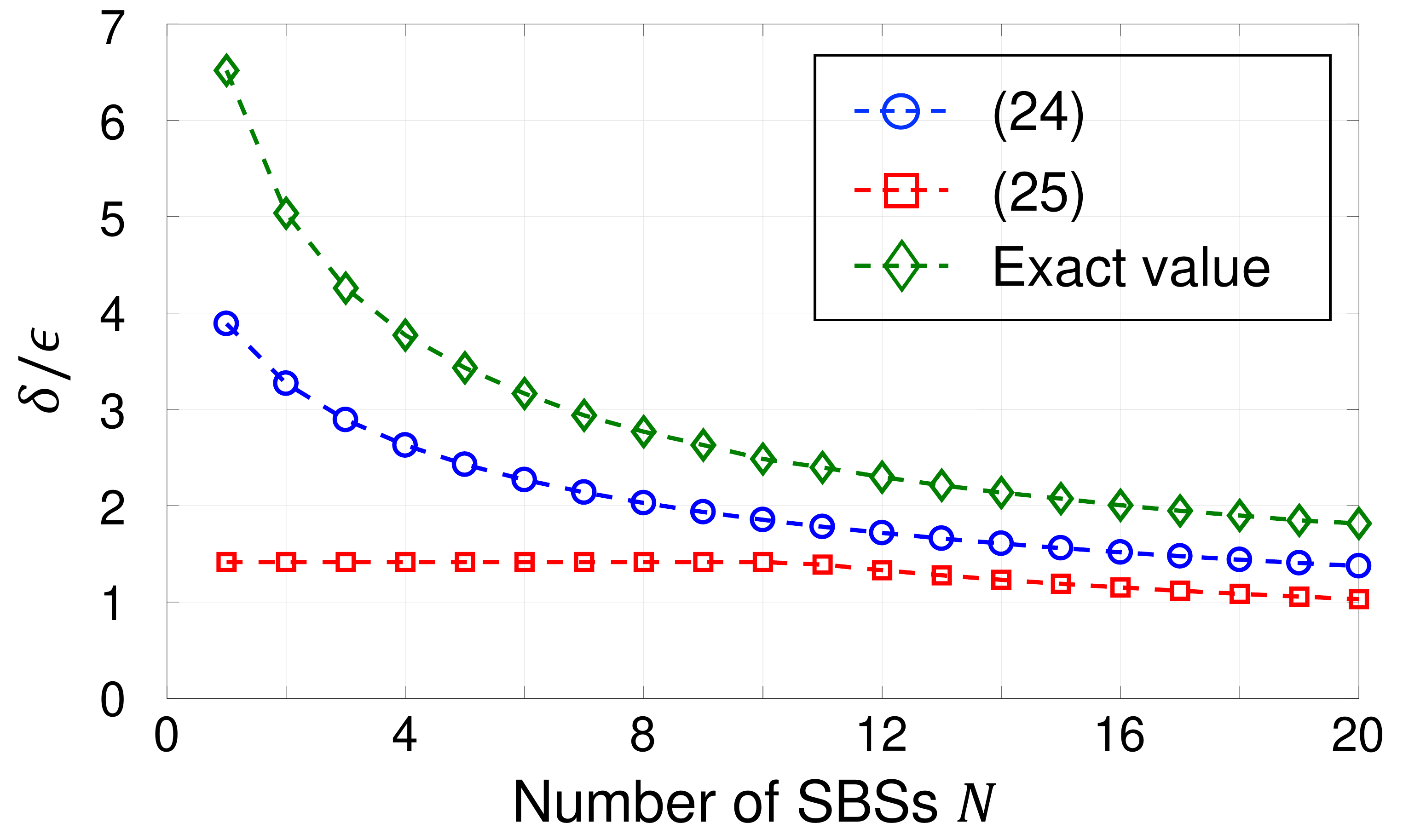}
	\caption{$\delta/\epsilon$ versus $N$.}
	\label{fig:delta_eps_plot}
\end{figure}

\begin{remark}
	Differential privacy has a limitation for intruders
	with knowledge of correlation between data when the data
	are continuously collected, as shown, e.g., in \cite{Cao2019}.
	To apply the above analysis in such a situation,
	the central controller needs to
	detect packet sniffing so quickly that
	intruders
	can get at most ``one-shot'' data.
	Another possible approach is to apply data-releasing mechanisms
	against attackers with knowledge of 
	data correlation proposed in
	\cite{Cao2019}, and we leave it for future studies.
\end{remark}

\section{Performance Evaluation}
\label{sec:numerical_simulation}
In this section, we illustrate the effectiveness of the distributed control method
through numerical simulations.
\subsection{Comparison with centralized control method}
\label{sec:comparison_simulation}
To evaluate robustness against masking noise,
we compare the distributed control method with
the centralized control method.
From this comparison, we also analyze 
errors of the approximation methods proposed in 
Section~\ref{sec:approximation_tech}.
Here we set  the number of the SBSs to $N=4$.
The reason for this small number of SBSs is 
that
the centralized control method solves 
the original problem \eqref{prob:original}.
The resulting worst-case complexity 
is at least $O(2^N)$, and therefore the centralized control method
is not applicable in the case of a large number of SBSs.
In contrast, 
the distributed control method is dimension-free as mentioned in
Remark~\ref{rem:complexity} and hence can deal with
a larger number of SBSs. A
simulation result of 
the distributed control method
with $N=16$ will be given in Section~\ref{sec:largeN}.
\subsubsection{Parameter settings}
The parameters for simulations are listed in Table~\ref{tab:parameters}.
The service area of every SBS is a circle with radius 0.4\,km, and hence 
the size is given by
$A = (0.4)^2 \pi$\,km$^2$.
Let 
the capacity and the initial energy of the battery in the SBSs be 
$X_{\max}=40$\,kJ and  $x_i[0] = 30$\,kJ, $i=1,\dots,4$, respectively.
The sampling period is given by $h = 300$\,s. 
The power amplifier efficiency is set to $\gamma = 0.32$. The MBS has
a maximum capacity of $U_{\text{Macro}} = 150$\,users.
The system powers are set to $s_{\text{active}} = 1.5$~W in the active mode and to 
$s_{\text{sleep}} = 0.5$\,W in the sleep mode.
The desired QoE value is given by $Q = 4$, and for a
data transmission model, the system noise, the packet size, and the path loss factor 
are set to $\sigma = -138.8$\,dBm, $S = 12000$\,bit (1500\,byte), and $Z=161.8296$\,dBm, respectively,
which are used in a standard LTE scenario studied in \cite{Suto2017}.

\begin{table}[tb]
	\centering
	\caption{Parameter settings.}
	\label{tab:parameters}
	\begin{tabular}{c|cc||c|c} \hline
		\textbf{Parameter} & \textbf{Value} && \textbf{Parameter} & \textbf{Value} \\ \hline 
		$N$ & 4 && $S$ & 12000\,bit 	 \\
		$A$ & $(0.4)^2\pi\,\text{km}^2$  && $\nu_{11}, \nu_{21}$ & 60, 70\,users\\
		$X_{\max}$ & 40\,kJ && 	$\nu_{12}, \nu_{22}$ & 90, 80\,users\\
		$x_i[0]$ & 30\,kJ & &$\nu_{13}, \nu_{23}$ & 70, 90\,users\\
		$\gamma$ & 0.32  &&$\nu_{14}, \nu_{24}$ & 80, 60\,users \\ 
		$h$ & 300\,s && $p_1, p_2$ & 144, 174 \\
		$U_{\text{Marco}}$ & 150\,users &&$p_3, p_4$ & 114, 144\\
		$s_{\text{acitve}}$ & 1.5\,W &&$\lambda $ & $5 \times 10^{-5} $\\
		$s_{\text{sleep}}$ & 0.5\,W  && $c$ & 0.1\\
		$\sigma$ & $-138.8$\,dB  &&$u_{\text{thres}}$ & 0.1\,W\\
		$Z$ & 161.8296\,dBm   &&$T$ &20\\
		$Q$ & 4  &&$\alpha^{[t]}$ & $7/t$ \\
		\hline
	\end{tabular}
\end{table}

We present two-day simulations of cell zooming.
The number of users of the area of the $i$th SBS is given by
\begin{equation}
\label{eq:user_num_sim}
U_i[k] =
\begin{cases}
\nu_{1i}  \exp\left( \frac{-(k-p_i)^2}{10^5} \right) & 
\text{if  $0\leq k < 288$}, \vspace{3pt}\\
\nu_{2i} \exp\left( \frac{-(k-p_i-288)^2}{10^5} \right)& 
\text{if $288 \leq k \leq 576$}
\end{cases}
\end{equation}
for each $i=1,\dots,4$. Note that since $h = 300$\,s, it follows that 
$288h$ is $24$ hours.
The numbers of users have the period of 24 hours and the peaks $\nu_{1i}$ and $\nu_{2i}$ at $k=p_i$ and $k=p_i+288$ on the first and
second days, respectively. These parameters are set to
$(\nu_{11},\nu_{21}, p_1) = (60,70,144)$, 
$(\nu_{12},\nu_{22}, p_2) = (90,80,174)$, 
$(\nu_{13},\nu_{23}, p_3) = (70,90,114)$, 
$(\nu_{14},\nu_{24}, p_4) = (80,60,144)$.
Fig.~\ref{fig:User_number} shows 
the number of users in the service area of each SBS.

The harvested power  is set to
\begin{equation}
\label{eq:harv_power_sim}
w_i[k] = 
\begin{cases}
10  \exp\left( \frac{-(k-144)^2}{5 \times 10^4} \right) & \text{if  $0\leq k < 288$}, \vspace{3pt}\\
10  \exp\left( \frac{-(k-432)^2}{5 \times 10^4} \right)& \text{if $288 \leq k \leq 576$}
\end{cases}
\end{equation}
for every $i=1,\dots,4$.
The harvested power has the period of 24 hours and the peak $10$\,W at noon, $k = 144$ and $k=432$.
The harvested power also has a bell-shaped curve similar to the numbers of users in Fig.~\ref{fig:User_number}.

\begin{figure}[tb]
	\centering
	\includegraphics[width = 7.5cm]{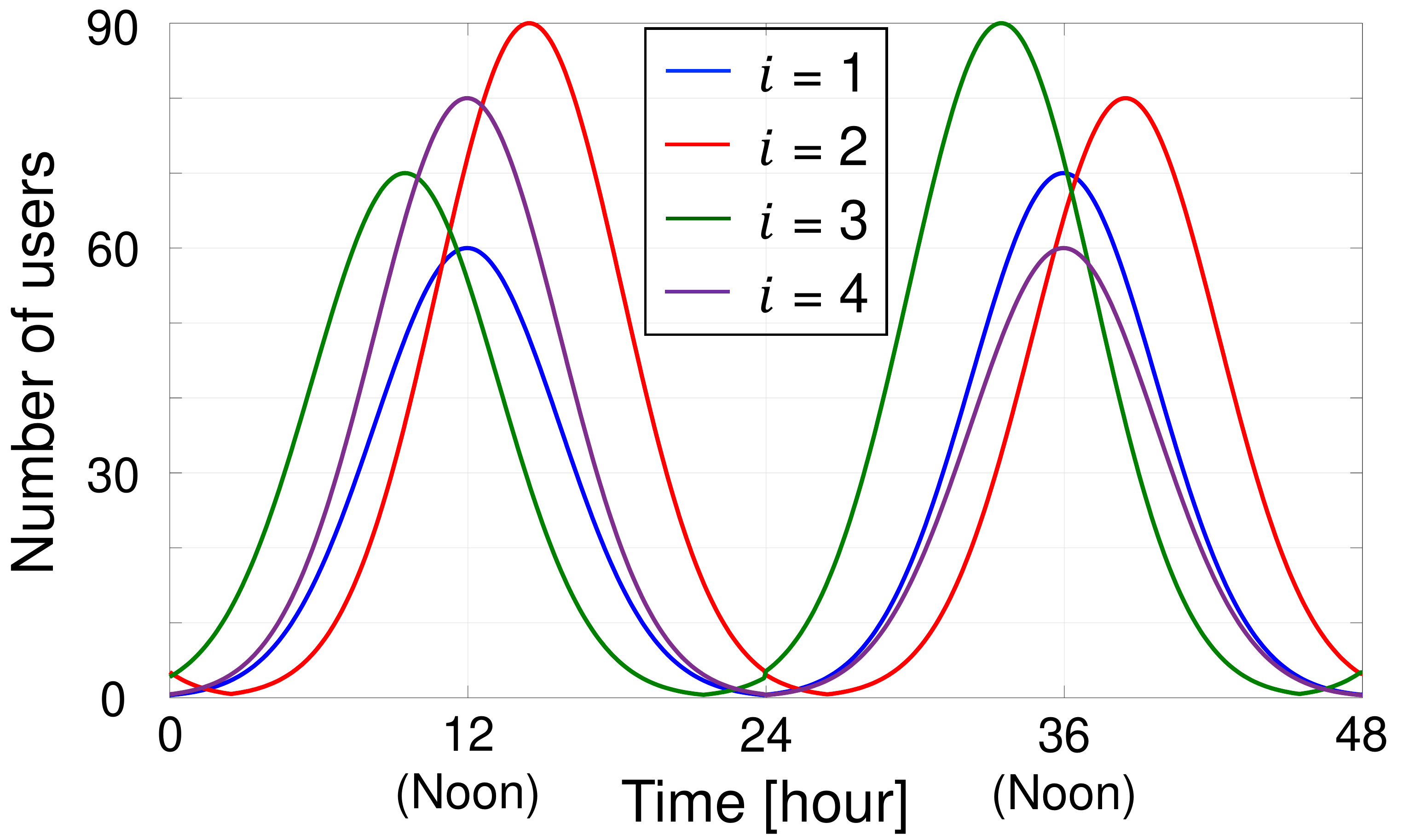}
	\caption{Numbers of users in service areas of SBSs.}
	\label{fig:User_number}
\end{figure}

The parameters of the proposed method are as follows.
The weight $\lambda$ of the objective function is
$\lambda = 5 \times 10^{-5}$.
The parameters for the approximation techniques are given by
$c = 0.1$ and $u_{\textrm{thres}} = 0.1$\,W.
The terminal step $T$ and the stepsize $\alpha^{[t]}$ of Algorithm~\ref{algo:DCZ}
are $T=20$ and $\alpha^{[t]} = 7/t$.

\subsubsection{Robustness of cell zooming against masking noise}
In this section, we compare the numbers of users associated with the SBSs in the presence/absence of
masking noise. This comparison verifies that 
the distributed control method, that is, Algorithm~\ref{algo:DCZ} and
the formula
\eqref{eq:app_solution},  is more robust against masking noise than
the centralized control method. 
Let each noise $v_i[k]$ be independently and identically 
distributed according the Laplace distribution  $\Lap(\rho)$. Hence,
for all $\epsilon,\delta>0$ satisfying $\delta/\epsilon = \rho$,
the cell zooming method achieves
$\epsilon$-differential privacy for 
$\delta$-adjacent pairs
by the Laplace mechanism.

Fig.~\ref{fig:Comparison} plots the numbers of users associated with the SBSs,
$\sum_{i=1}^{N} F\big(u_i[k],U_i[k]\big)$,
in the noiseless case $\rho=0$ and the case $\rho=10$.
In Fig.~\ref{fig:Comparison_SOC}, we present
the residual energy of the battery in each SBS in the case $\rho=10$.
The centralized control case is given in Figs.~\ref{fig:SBS_users_NA} and \ref{fig:SBS_SoC_NA}, and
the distributed control case is in  Figs.~\ref{fig:SBS_users} and \ref{fig:SBS_SoC}.

From Figs.~\ref{fig:Comparison} and \ref{fig:Comparison_SOC}, we observe that
the centralized control method
is sensitive to masking noise, whereas  in the distributed control case, 
the effect of masking noise appears only in short periods such as the interval $[31, 34]$.
This is because in the distributed control method, 
the objective function is not affected by the masking noise.
Although the constraint [C5] is affected by 
the noise even in the distributed control case, 
this constraint
is satisfied for all sufficiently large transmission powers.
Moreover, the analysis in Section~\ref{sec:constraint_error} shows that 
the effect of the masking noise to the constraint becomes smaller
as the transmission power increases.
Hence the  error of [C5]  due to the masking noise 
does not change the optimal solution in the situation where
the SBSs have plenty of energy.
In the simulation of the distributed control case, the available 
energy of two SBSs $i=2,3$ is small in the interval $[31,34]$ as shown in Fig.~\ref{fig:SBS_SoC},
and hence we observe the noise effect in this interval in
Fig.~\ref{fig:SBS_users}.
The residual energy of 
the SBS $i=2$ is already depleted before the interval $[31,34]$. However,
its effect is not seen 
in Fig.~\ref{fig:SBS_users}, because 
only a few users are active in the area of the SBS $i=2$ in the interval $[24,31]$ as shown in Fig.~\ref{fig:User_number}.

Moreover, Fig.~\ref{fig:Comparison} shows that in the noiseless case,
the number of users by the distributed control method
is almost the same as that by the centralized control method, although
the approximate problem is solved in the distributed control case.
This implies that the proposed approximation techniques have
only small errors.

\begin{figure}
	\centering
	\subcaptionbox{Centralized control.\vspace{10pt} 
		\label{fig:SBS_users_NA}}
	{\includegraphics[width = 7.5cm,clip]{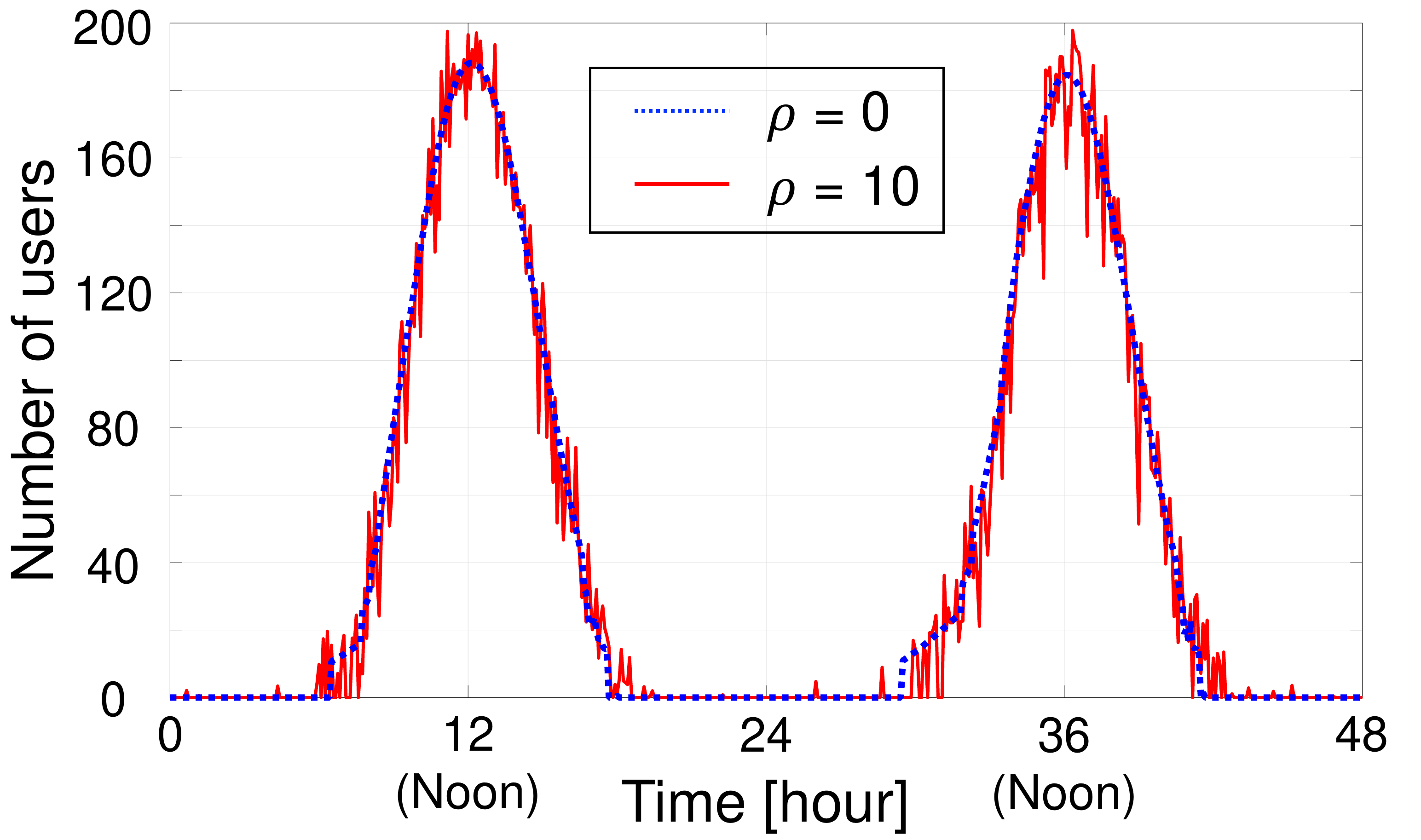}} 
	\subcaptionbox{Distributed control.
		\label{fig:SBS_users}}
	{\includegraphics[width = 7.5cm,clip]{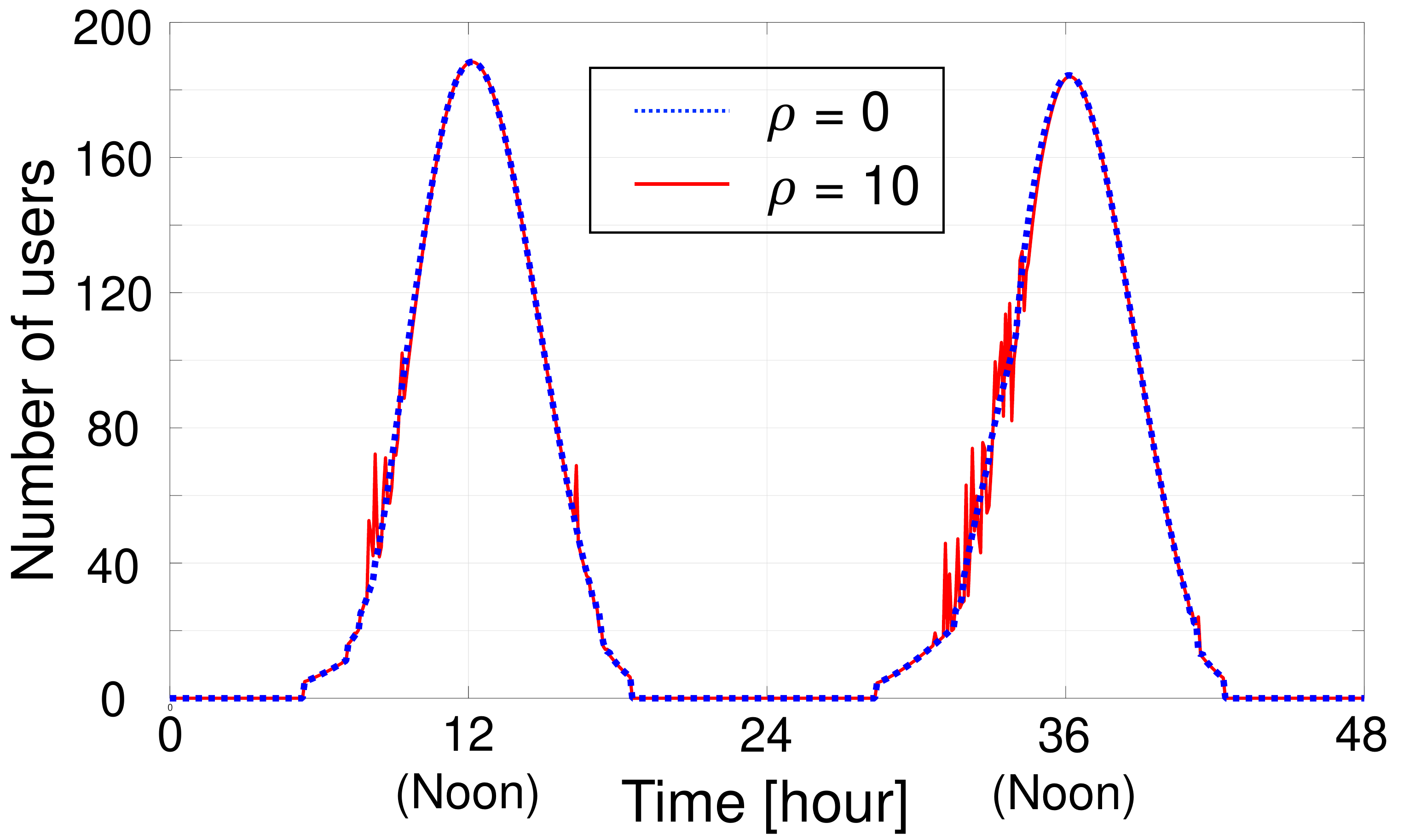}}
	\caption{Numbers of users associated with SBSs
		in the cases $\rho = 0, 10$. \label{fig:Comparison}}
\end{figure}

\begin{figure}
	\centering
	\subcaptionbox{Centralized control.\vspace{10pt} 
		\label{fig:SBS_SoC_NA}}
	{\includegraphics[width = 7.5cm,clip]{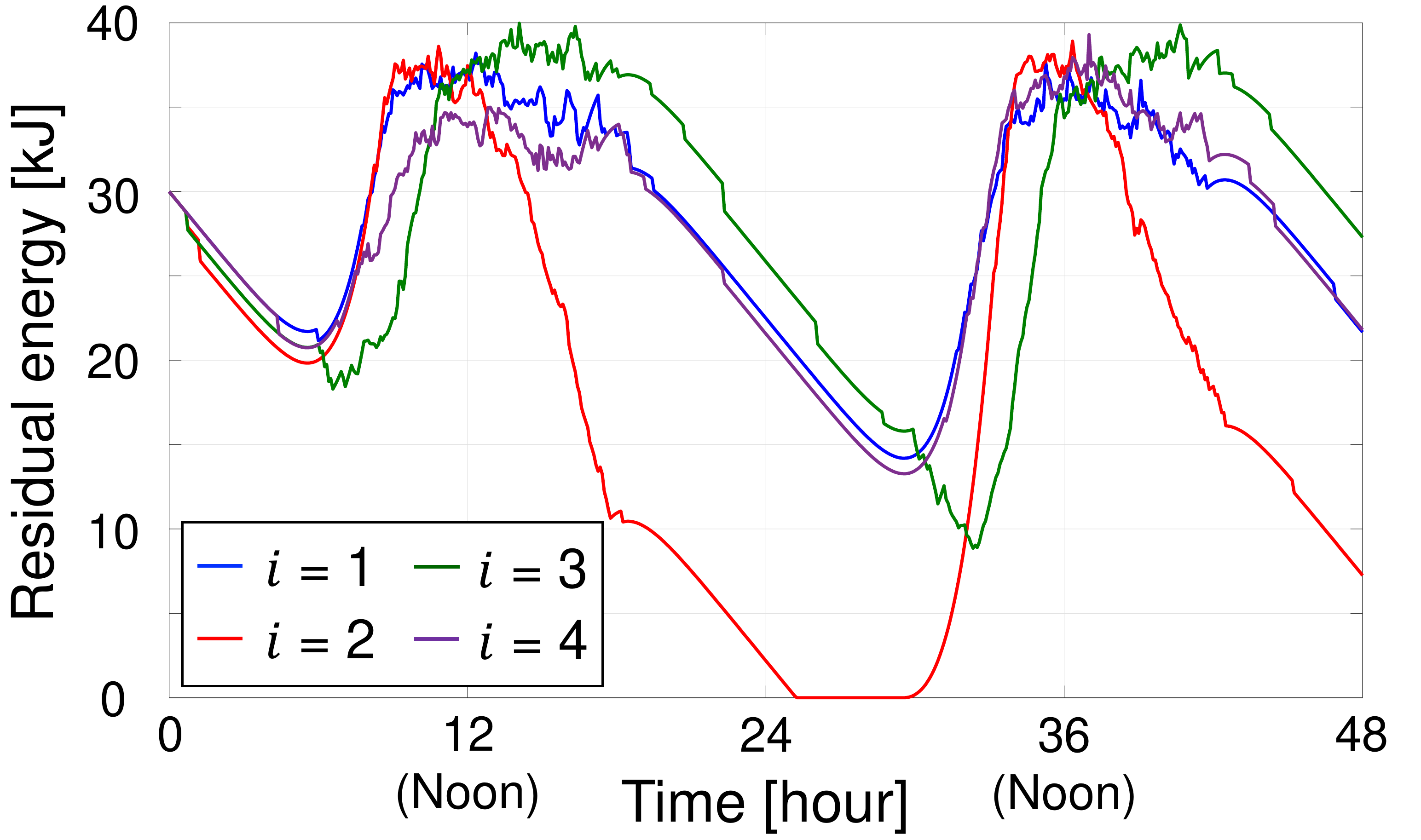}} 
	\subcaptionbox{Distributed control.
		\label{fig:SBS_SoC}}
	{\includegraphics[width = 7.5cm,clip]{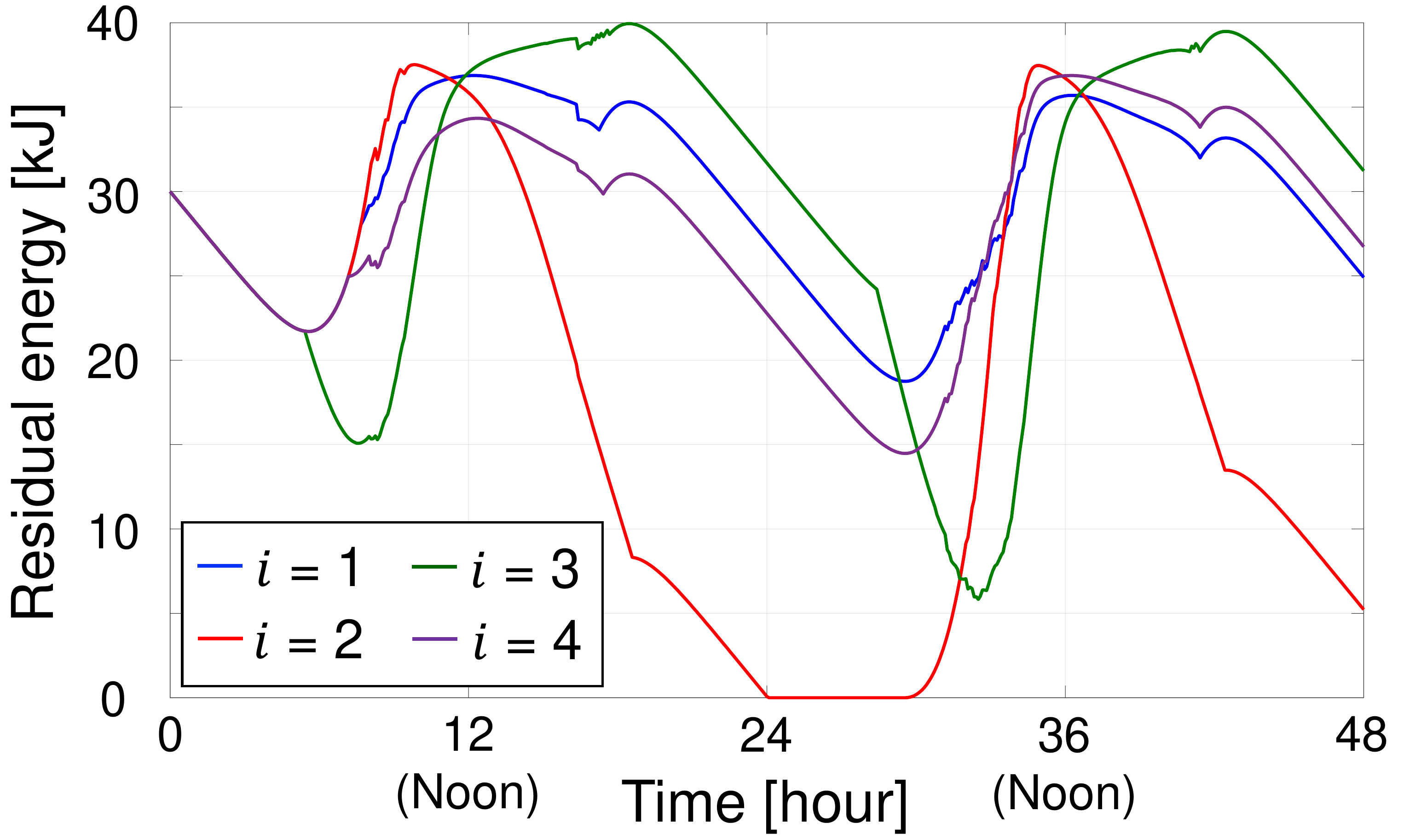}}
	\caption{Residual energy of each SBS in the case $\rho =10$. \label{fig:Comparison_SOC}}
\end{figure}

Fig.~\ref{fig:effect_noise} plots
control performances
versus noise intensity.
As control performances, 
we consider  
the energy efficiency over the two days, i.e.,
\begin{equation}
\frac{\text{(Total number of users associated with SBSs)}}
{\text{(Total consumed  energy [kJ])}}
\end{equation}
in Fig.~\ref{fig:EE}
and  the average of the charging rate of the
batteries, i.e.,
\begin{equation}
\frac{1}{577N}\sum_{k=0}^{576} \sum_{i=1}^N \frac{x_i[k]}{X_{\max}}
\end{equation}
in Fig.~\ref{fig:CR}.
We compute the above criteria
with $500$ samples of ${\bf v} \sim \Lap(\rho)^N$ for both the centralized and distributed control methods.
As expected, in the noiseless 
case $\rho = 0$, the performance of distributed control method
is slightly worse than that of the centralized control method due to
approximation errors. However, we see from Fig.~\ref{fig:effect_noise}
that the performance of the centralized control method
becomes worse sharply as the intensity of masking noise increases.
In contrast, the distributed control method is robust against masking noise, and 
in particular, its performance is not harmed by  masking noise with small intensity.

\begin{figure}
	\centering
	\subcaptionbox{Energy efficiency versus noise intensity.\vspace{10pt} 
		\label{fig:EE}}
	{\includegraphics[width = 7.5cm,clip]{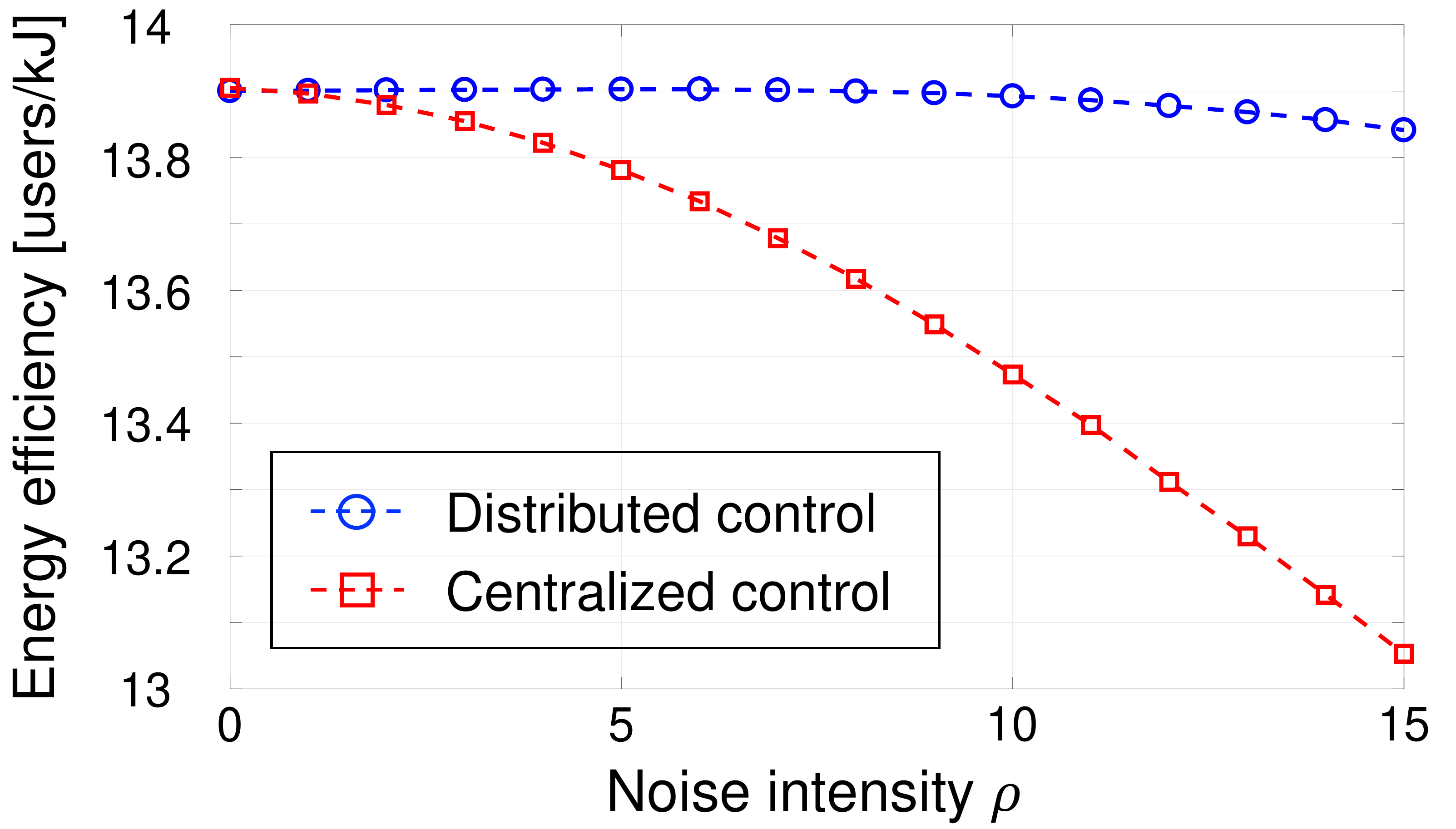}} 
	\subcaptionbox{Charging rate versus noise intensity.
		\label{fig:CR}}
	{\includegraphics[width = 7.3cm,clip]{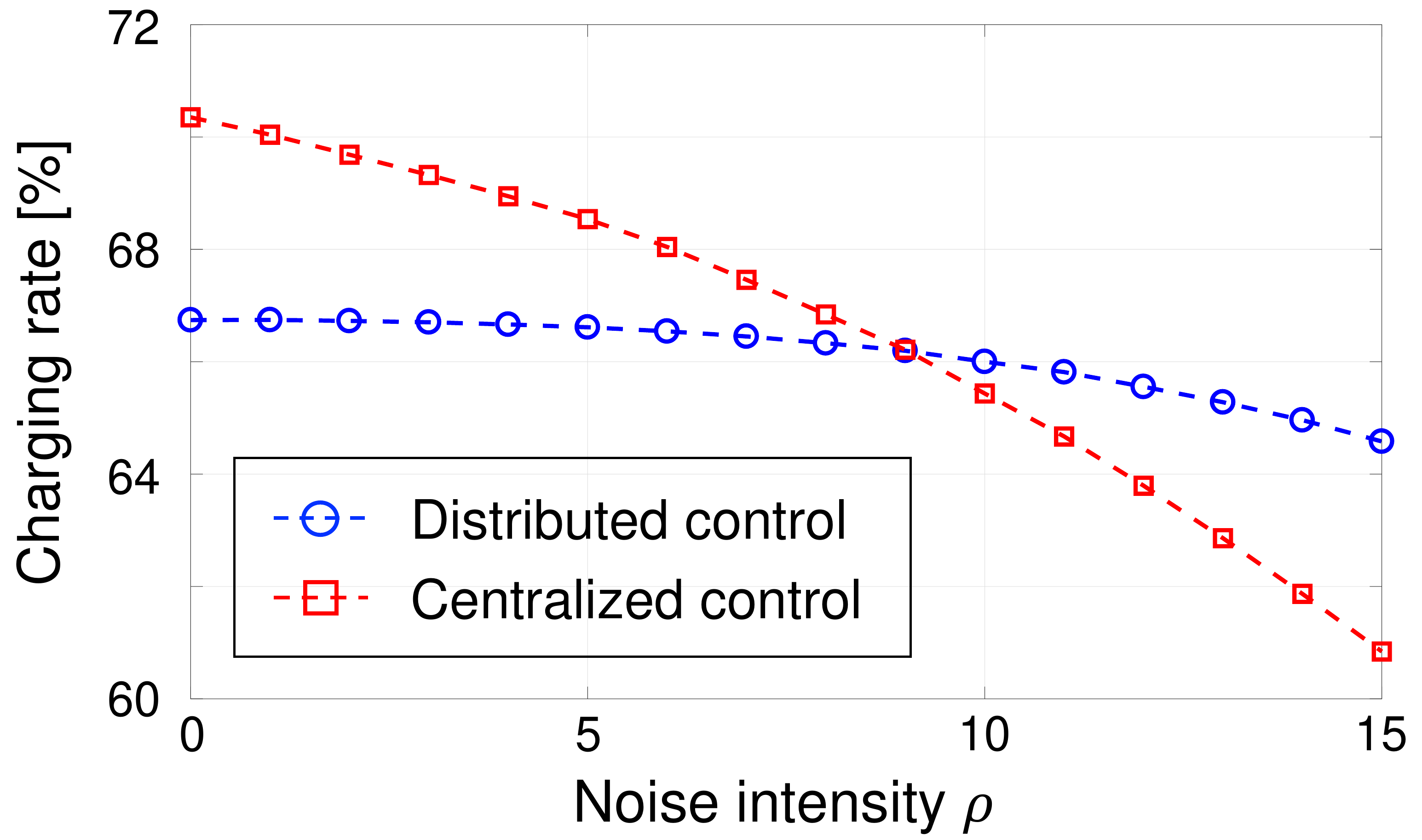}}
	\caption{Robustness against noise.
		For all $\epsilon,\delta>0$ satisfying $\delta/\epsilon = \rho$,
		the cell zooming method achieves $\epsilon$-differential privacy 
		for $\delta$-adjacent pairs by the Laplace mechanism.
		\label{fig:effect_noise}}
\end{figure}

\subsubsection{Truncation and approximation errors}
In this section, we investigate the following errors:
\begin{itemize}
	\item
	the truncation error of Algorithm~\ref{algo:DCZ} due to
	the finiteness of the terminal step $T$;
	\item
	the approximation error due to the transformation of the 
	original minimization problem \eqref{prob:original}  into the form of \eqref{eq:app_min_pro2}.
\end{itemize}
In the simulations, we consider the noiseless case, that is, $v_i[k] = 0$ for every $i,k$.

Fig.~\ref{fig:Communication_num} depicts the percent error due to truncation,
i.e.,
\begin{equation}
100 \times \frac{\sum_{i=1}^N\big\|u_i^{[T]} - u_i^{[30]}\big\|_2 ~\!/~\! \big\|u_i^{[30]}\big\|_2 }{N},
\end{equation}
for $1\leq T \leq 29$,
where 
\[
\|u\|_2 := \sqrt{\sum_{k=0}^{576} \big|u[k]\big|^2}
\]
 for $u = \begin{bmatrix}
u[0] & \cdots & u[576]
\end{bmatrix}$ and
$u_i^{[T]}$ is the transmission power $u_i$  obtained 
by Algorithm~\ref{algo:DCZ} and 
the formula \eqref{eq:app_solution} with the terminal step $T$.
The transmission power computed in Algorithm~\ref{algo:DCZ}
converges as $T \to \infty$, 
but this does not imply that 
the truncation error is monotonically decreasing with respect to $T$.
Hence the truncation error temporarily increases at $T=5$
in Fig.~\ref{fig:Communication_num}.
We see from Fig.~\ref{fig:Communication_num} that for $T \geq 10$, 
the truncation error is sufficiently small.
Since the data transfer time between SBSs and MBSs  is at most
a few milliseconds, 
the proposed method computes the transmission power sufficiently fast.

\begin{figure}[tb]
	\centering
	\includegraphics[width = 7.5cm]{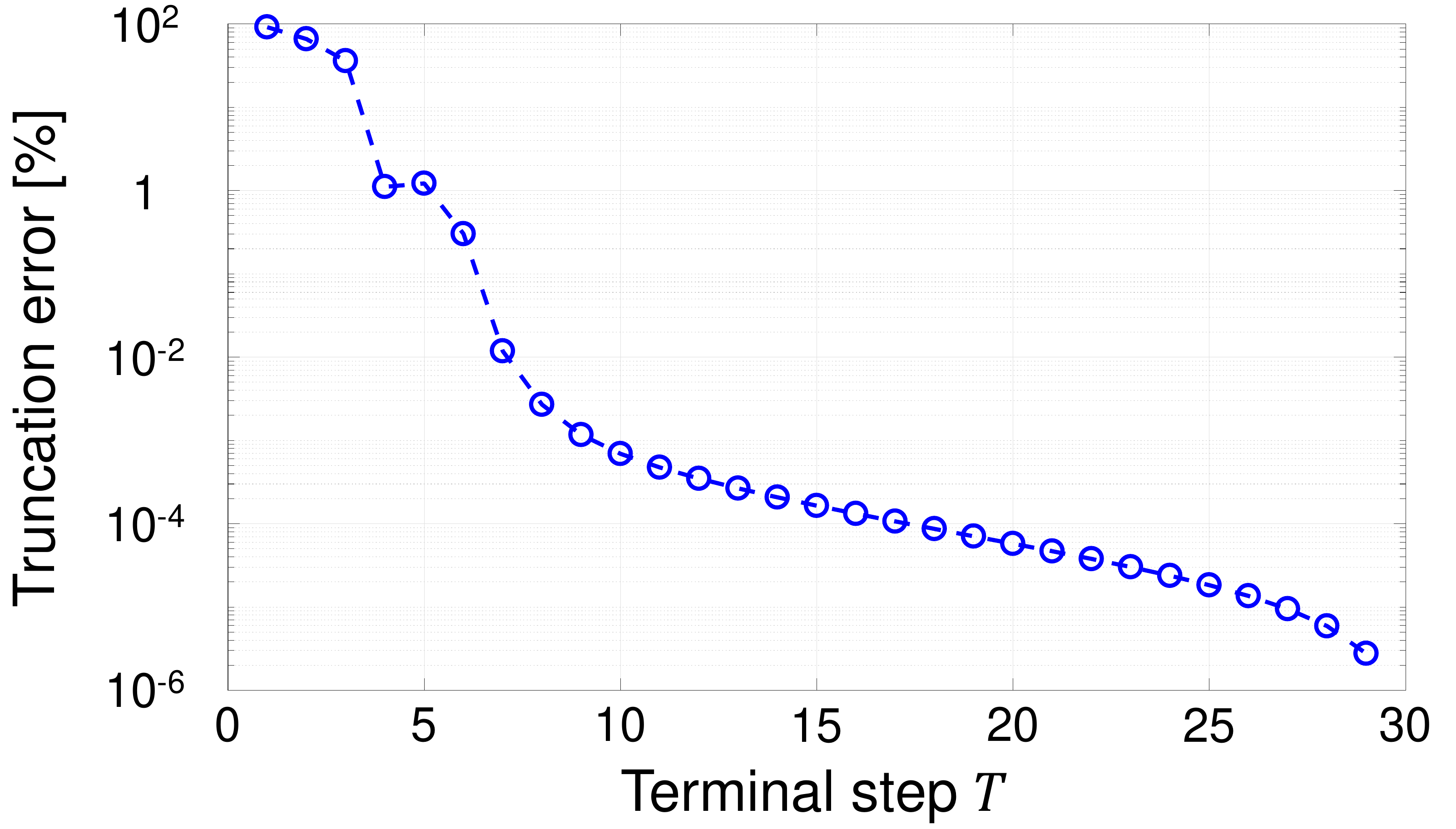}
	\caption{Truncation error versus terminal step $T$.}
	\label{fig:Communication_num}
\end{figure}

\begin{figure}[tb]
	\centering
	\includegraphics[width = 7.1cm]{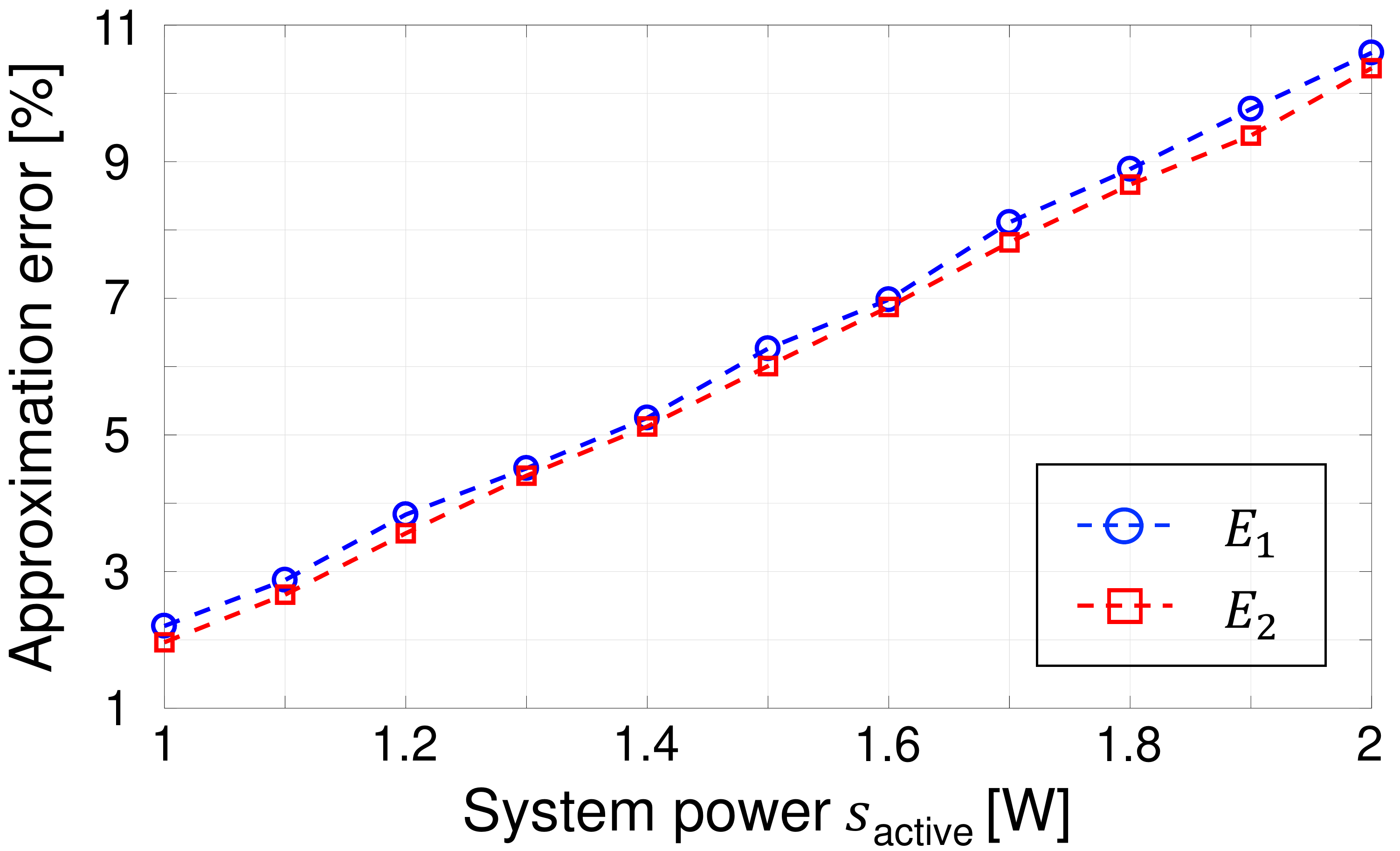}
	\caption{Approximation error versus system power $s_{\text{active}}$.}
	\label{fig:S_active}
\end{figure}
Fig.~\ref{fig:S_active} plots the approximation error
versus the system power in the active mode $s_{\text{active}}$.
Let $u_i^{\text{exact}}$ be the 
transmission power of the $i$th SBS obtained by
solving the minimization problem~\eqref{prob:original}.
Let $u_i^{\text{app}1}$ and $u_i^{\text{app}2}$ be the 
transmission powers of the $i$th SBS 
obtained by the Algorithm~\ref{algo:DCZ}.
To obtain $u_i^{\text{app}1}$,
we use the formula \eqref{eq:app_solution} for the approximate solution
of the minimization problem \eqref{eq:local_minimization_prob}, while
we compute the exact solution of the problem  \eqref{eq:local_minimization_prob}
for $u_i^{\text{app}2}$.
The circles and the squires in Fig.~\ref{fig:S_active} show the percent errors due to approximation, i.e,
\begin{align}
E_l := 100  \times \frac{\sum_{i=1}^N\big\|u_i^{\text{app} \hspace{1pt}l} - u_i^{\text{exact}}\big\|_2 ~\!/~\! \big\|u_i^{\text{exact}}\big\|_2}{N}
\end{align}
for $l =1,2$, respectively.
Fig.~\ref{fig:S_active} indicates that 
the errors of the approximation techniques 
developed in Section~\ref{sec:approximation_tech}
are quite small. 
We also observe that the approximation error becomes 
larger as $s_{\text{active}}$
increases. This is because 
the $\ell^1$ approximation of the $\ell^0$ norm  in 
\eqref{eq:l0norm_app} yields
larger errors as the difference of the system powers
between the active and sleep modes, $s_{\text{active}}^2 - s_{\text{sleep}}^2$, increases.
Moreover, 
$u_i^{\text{app1}}$ has a slightly larger approximation error than
$u_i^{\text{app2}}$, but  the difference is  negligibly small.

\subsection{Simulation of a large-sized cell network}
\label{sec:largeN}
Here we present a simulation result of
a large-sized cell network.
The number of the SBSs and
the maximum capacity of the MBS are set to
$N =16$ and $U_{\text{Macro}} = 400$, respectively.
The number of active users in the service area of each SBS is of the form
\eqref{eq:user_num_sim}, where
the peaks $\nu_{1i}$, $\nu_{2i}$ 
and the peak time $p_i$ are chosen from the 
discrete uniform distribution on the sets 
$\{40,41,\dots,70\}$ and $\{114,115, \dots,174\}$,
respectively.
The harvested power is given by
\eqref{eq:harv_power_sim} for every SBS.
Other parameters are the same as in Section~\ref{sec:comparison_simulation}.

Fig.~\ref{fig:N16} shows the  numbers of users accommodated 
by the SBSs, $\sum_{i=1}^{N} F\big(u_i[k],U_i[k]\big)$, for
the noise intensities $\rho = 0,10$.
As in the case $N=4$, we observe the robustness of the proposed cell zooming
method to the masking noise. The number of users 
fluctuates 
in some intervals in Fig.~\ref{fig:N16}.
The reason   is the same as in the case $N=4$. In fact,
the remaining energy of some SBSs is small in the intervals such as $[31,34]$.
Since the simulation result of the SBS energy shows 
a trend almost identical to
the case $N=4$, we omit it. 
\begin{figure}[tb]
	\centering
	\includegraphics[width = 7.5cm]{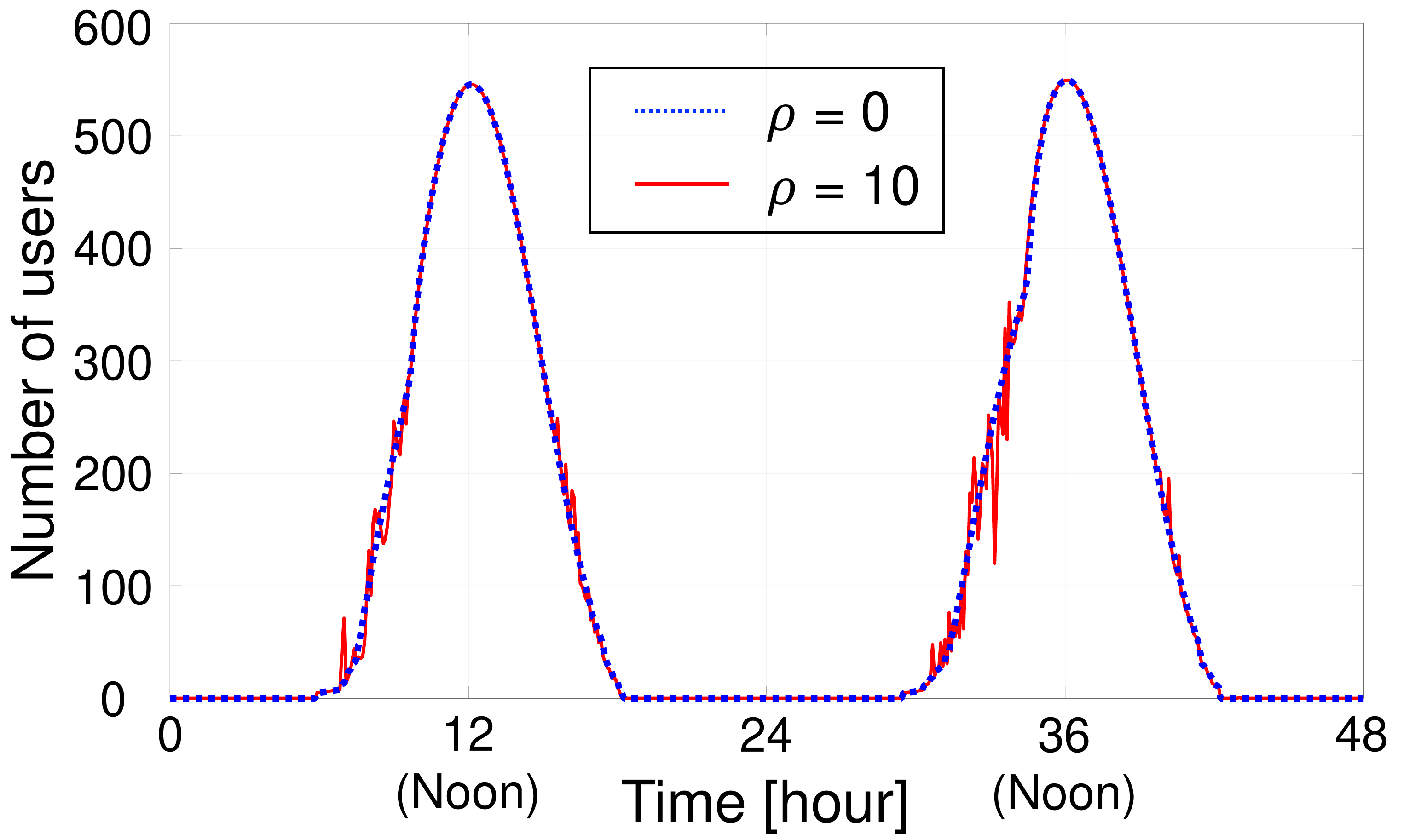}
	\caption{Numbers of users associated with SBSs in the case $N=16$.}
	\label{fig:N16}
\end{figure}

\section{Conclusion}
\label{sec:conclusion}
We have proposed a cell zooming method with masked data for off-grid SBSs.
We have formulated the minimization problem of cell zooming, in which
the number of users associated with the SBSs and the available energy of the batteries in the SBSs
are evaluated.
To solve the minimization problem, the measurement data on the numbers of users in 
the service areas of the SBSs are required. We have preserved the confidential measurement data, by
adding masking noise to them.
We have developed a distributed cell zooming algorithm that is more robust to masking noise than
the conventional centralized method.
Although the originally formulated minimization problem is equivalent to
a mixed integer nonlinear programing problem, the proposed algorithm 
computes its approximate solution with low computational complexity.
In addition, we have analyzed the trade-off between
confidentiality  and optimization accuracy
by using the notion of differential privacy.
Our numerical results have  shown that the distributed control 
method outperforms
the centralized control 
method with respect to robustness against masking noise.
Moreover, we have observed that the approximation error of the proposed method
is quite small.

\appendices
\renewcommand{\thetheorem}{\Alph{theorem}}
\section{Derivation of (9)}
\label{sec:appA}
First, we approximate the objective function $P_k$ given in \eqref{eq:obj_func}.
The approximation techniques in 1)\hspace{1pt}--\hspace{1pt}3) of Section~\ref{sec:approximation_tech} yield
\begin{align}
&\sum_{i=1}^N\Big(
X_{\max} - 
\sat
\big(
x_i[k] - h(w_i[k]-u_i/\gamma - s_i)
\big)
\Big)^{2} \notag \\
&~\approx
\sum_{i=1}^N \Big\{(X_{\max} - 
x_i[k] - hw_i[k] + h u_i/\gamma)^2 \\
&\hspace{38pt}  +h^2 \big(s_{\rm active}^2 - s_{\rm sleep}^2\big) u_i + 
h^2 s_{\rm sleep}^2 \notag \\
&\hspace{38pt} + 
2h(X_{\max} - 
x_i[k] - hw_i[k] + h u_i/\gamma)s_i[k-1]\Big\}. \notag 
\end{align}
Therefore, the objective function $P_k$ can be approximated as
\begin{align}
&P_k({\bf u}, {\bf U}) \approx
\sum_{i=1}^{N}
\Big\{
f_{i,k}\big(u_i,U_i[k]\big) + h^2 s_{\rm sleep}^2 \notag \\
&\qquad ~~  + 2\lambda h\big(X_{\max} - 
x_i[k] - hw_i[k] \big)s_i[k-1]
\Big\},
\label{eq:Pk_app}
\end{align}
where the function $f_{i,k}$ is defined by \eqref{eq:fik}.
Since the terms 
\begin{equation}
h^2 s_{\rm sleep}^2,\quad 
2\lambda h\big(X_{\max} - 
x_i[k] - hw_i[k] \big)s_i[k-1]
\end{equation}
do not depend on the variables $u_1,\dots,u_N$,
it follows that minimizing the function in the right-hand side of \eqref{eq:Pk_app} 
is equivalent to minimizing 
$\sum_{i=1}^{N}f_{i,k}(u_i,U_i[k]) $.

Next, we investigate the constraints [C3] and [C4].
The constraint [C3] is equivalent to
\begin{equation}
\label{eq:u_i_upper_bound}
u_i \leq \gamma (
x_i[k]/h + w_i[k] - s_i
).
\end{equation}
Since $u_i \geq 0$ and since $s_i = s_{\text{sleep}}$ is equivalent to $u_i = 0$,
we can rewrite \eqref{eq:u_i_upper_bound} as
\begin{equation}
u_i \leq \max \big\{ 0,~\gamma (
x_i[k]/h + w_i[k] - s_{\text{active}} ) \big\}.
\end{equation}
Moreover, 
for $U_i[k] >0$,
the constraint [C4] is equivalent to
$u_i \leq r^{-\frac{10}{19}}.$
Hence, the constraints [C3] and [C4] can be reduced to
$u_i \leq u_i^{\max}[k]$, where $u_i^{\max}[k]$ is defined by \eqref{eq:umax}.

Finally, if we define a function $g$ as in \eqref{eq:gik},
then $\sum_{i=1}^N g(u_i,U_i[k]) \leq 0$ if and only if
[C5'] given in \eqref{eq:margin_for_FA} holds. 
Thus, 	
the minimization problem \eqref{prob:original}  is approximated by
a minimization problem
of the form \eqref{eq:app_min_pro2}.

\section{Proof of Theorem~\ref{thm:third_order}}
\label{sec:appB}
a)
Define the nonlinear function $\Xi$ by
\begin{equation}
\label{eq:theta_def}
\Xi(u) := p_1u^{\frac{28}{19}}
+p_2 u^{\frac{10}{19}} + 
p_3 u^{\frac{9}{19}} - p_4.
\end{equation}
We will show that $\Xi(u) = 0$ has a unique positive solution.
Note that the coefficients $p_1$, $p_2$, and $p_4$ are non-negative but that
$p_3$ may be negative.
Let us first consider the case $p_3 \geq 0$.
In this case,  $\Xi$ is monotonically increasing.
Since $\Xi(0) = -p_4 < 0$,
it follows that $\Xi(u) = 0$ has a unique positive solution $u^{\sol}$.

Next, we consider the case $p_3 <0$.
The derivative $\Xi'(u)$ is given by
\begin{equation}
\Xi'(u) = \frac{28p_1u + 10p_2 u^{\frac{1}{19}} + 9p_3}{19u^{\frac{10}{19}}}.
\end{equation}
There exists a unique positive solution $u = \tilde u$ 
of $\Xi'(u) = 0$, and $\Xi$ has a minimum at $u = \tilde u$.
Hence, 
there uniquely exists $u^{\sol}>0$ such that 
$\Xi(u^{\sol}) = 0$.

We prove that the solution $u^*$ of the 
minimization problem \eqref{eq:local_minimization_prob}
is given by \eqref{eq:u_opt}.
Define the function $L(u)$ by
\begin{align}
L(u) &:= f(u,U) + \mu g(u,U+v) \notag \\
&= R_1\left(1-r u^{\frac{10}{19}}\right)^2 + 
R_2\left(1-r u^{\frac{10}{19}}\right) \\
&\qquad + 
\lambda (R_3+\beta_1 u)^2 + \beta_2 u -  V, \notag 
\end{align}
where 
\begin{align}
R_1 &:= U^2,\quad R_2 := \mu (U+v),\quad R_3 :=
X_{\max} - x - hw  \\
\beta_1 &:= \frac{h}{\gamma},\quad 
\beta_2 := \lambda h^2
\left(
s_{\rm active}^2 - s_{\rm sleep}^2 + \frac{2s[k-1]}{\gamma}
\right)  \\
V &:= \frac{\mu(1-c)U_{\text{Macro}}}{N}.
\end{align}
Then we obtain $L'(u) = u^{-\frac{9}{19}} \Xi(u)$,
where $\Xi$ is defined as in \eqref{eq:theta_def}.
Since $\Xi(u) < 0$ for all $u<u^{\sol}$ and $\Xi (u) > 0$ for all $u > u^{\sol}$,
it follows that 
$
\min_{u\geq 0} L(u) = L(u^{\sol}).
$
If $u^{\sol} > u^{\max}$, then $L$ is monotonically decreasing on $[0,u^{\max}]$.
Thus,
the solution $u^*$ 
of the minimization problem \eqref{eq:local_minimization_prob} is given by
\eqref{eq:u_opt}.

b) The second assertion follows from an argument similar to that in 
the proof of Theorem~3 of \cite{Wakaiki2018EHSCN}.
Suppose that $p_3 \geq 0$. Define
\begin{align}
\Xi_1(u) &:=  p_1u^{\frac{27}{19}}
+p_2 u^{\frac{9}{19}} + 
p_3 u^{\frac{9}{19}} - p_4\\
\Xi_2(u) &:=  p_1u^{\frac{30}{19}}
+p_2 u^{\frac{10}{19}} + 
p_3 u^{\frac{10}{19}} - p_4.
\end{align}
Note that $\Xi_1$ and $\Xi_2$ are in the form of a
cubic polynomial 
\[
p_1\chi^3 + (p_2+p_3) \chi - p_4
\]
with  $\chi = u^{\frac{9}{19}}$ and 
$\chi = u^{\frac{10}{19}}$, respectively.
Since $\Xi_1(0) = \Xi_2(0) = -p_4 < 0$ and since 
$\Xi_1$ and $\Xi_2$ are monotonically increasing,
$\Xi_1(u) = 0$ and $\Xi_2(u) = 0$ have unique positive solutions $u_1^{\sol}$ and $u_2^{\sol}$,
respectively. 
We obtain
\begin{align}
\Xi(u) - \Xi_1(u) &= 
p_1 \Big(
u^{\frac{28}{19}} - u^{\frac{27}{19}}
\Big) + 
p_2 \Big(
u^{\frac{10}{19}} - u^{\frac{9}{19}}
\Big)\\
\Xi_2(u) - \Xi(u)  &= 
p_1 \Big(
u^{\frac{30}{19}} - u^{\frac{28}{19}}
\Big) + 
p_3 \Big(
u^{\frac{10}{19}} - u^{\frac{9}{19}}
\Big),
\end{align}
and it follows from $p_1 >0$ and $p_2,p_3 \geq 0$ that 
\begin{align}
\begin{cases}
\Xi_2(u) < \Xi(u) < \Xi_1(u) & \text{if $0< u<1$},\\
\Xi_1(u) < \Xi(u) < \Xi_2(u) & \text{if $u >1$.}
\end{cases}
\end{align}
Moreover, $\Xi_1(1) = \Xi(1) = \Xi_2(1)$.
Hence
\begin{equation}
\label{eq:v^*_v1_v2}
\begin{cases}
u_1^{\sol} \leq u^{\sol} \leq u_2^{\sol} & \text{if $u_2^{\sol} \leq 1$}, \\
u_2^{\sol} \leq u^{\sol} \leq u_1^{\sol} & \text{if $u_2^{\sol} > 1$.}
\end{cases}
\end{equation}	
The unique positive solution  of the
cubic equation 
\[
p_1\chi^3 + (p_2+p_3) \chi - p_4 = 0
\]
is given by \eqref{eq:third_order_eq_solution}.
Together with this, \eqref{eq:v^*_v1_v2} yields \eqref{eq:v*_app}.
\hspace*{\fill} $\blacksquare$

\section{Proof of Proposition~\ref{prop:diff_privacy}}
\label{sec:appC}
We use
a tail bound for the sum of independent and identically distributed Laplace random variables 
in the proof of Proposition~\ref{prop:diff_privacy}.
The following tail bound can be obtained by a standard technique based on the
Chernoff bound  (see, e.g., \cite[Section~1.2]{Rigollet2017}), but we give the proof for completeness.
\begin{lemma}
	\label{lem:laplace_tail_prob}
	For all $\rho >0$ and $\Lambda >0$,
	independent and identically distributed random variables $W_1,\dots,W_N \sim \Lap(\rho)$
	satisfy
	\begin{equation}
	\label{eq:laplace_summation}
	\textrm{{\em P}}	
	\left(
	\left|\sum_{i=1}^N W_i \right| > \Lambda
	\right) \leq 2\psi_N(\Lambda/\rho),
	\end{equation}	
	where $\psi_N$ is defined by \eqref{eq:psi_definition}.
\end{lemma}

\begin{proof}
	It is enough to show that \eqref{eq:laplace_summation} holds for $\rho = 1$.
	We obtain the general case $\rho \not= 1$ by
	replacing $W_i$ and $\Lambda$
	by
	$W_i/\rho \sim \Lap(1)$ and $\Lambda/\rho$, respectively.

	For each $i=1,\dots,N$,
	the moment generating function of $W_i\sim \Lap(1)$ is given by
	\begin{equation}
	\textrm{E}	\big[e^{W_i z}\big] = \frac{1}{1-z^2}\qquad \forall z \in (-1,1).
	\end{equation}
	Using the Chernoff bound, we obtain
	\begin{align}
	\textrm{P}	\left(
	\sum_{i=1}^N W_i  > \Lambda
	\right)  &\leq 
	e^{-\Lambda z} \prod_{i=1}^N \textrm{E}	\big[e^{W_i z}\big] \\
	&= 
	\frac{e^{-\Lambda z}}{(1-z^2)^N} =: \phi (z)
	\end{align}
	for every $z\in (0,1)$.
	
	We next investigate $\inf_{0<z<1} \phi(z)$.
	Define 
	\[
	\Phi(z) := \ln \phi(z) = -\Lambda z - N \ln \big(1-z^2\big). 
	\]
	Since
	\begin{equation}
	\Phi'(z) = -\Lambda  + \frac{2Nz}{1-z^2},
	\end{equation}
	it follows that $\inf_{0<z<1} \Phi(z) = \Phi\big(\varpi(\Lambda )\big)$, where
	\begin{equation}
	\varpi(\Lambda ) := \frac{\sqrt{N^2+\Lambda ^2} - N}{\Lambda } \in (0,1)\qquad \forall N \in \mathbb{N}.
	\end{equation}
	The increasing property of logarithmic functions yields
	\begin{equation}
	\inf_{0<z<1} \phi(z) = \phi\big(\varpi(\Lambda ) \big) = \psi_N(\Lambda ),
	\end{equation}
	where $\psi_N$ is defined by \eqref{eq:psi_definition}.
	We obtain the same bound for 
	\[
	\textrm{P}	\left(
	\sum_{i=1}^N W_i  < -\Lambda 
	\right).
	\]
	Thus, \eqref{eq:laplace_summation} holds.
\end{proof}

{\em Proof of Proposition~\ref{prop:diff_privacy}:}
By Lemma~\ref{lem:laplace_tail_prob}, 
we obtain \[
\textrm{P}	
\left(
\left|\sum_{i=1}^N W_i \right| > \Lambda 
\right) < \zeta
\]
 for 
$W_i \sim \Lap(\delta/\epsilon )$, $1\leq i \leq N$, if $2\psi_N(\epsilon \Lambda /\delta) < \zeta$. 
The $\epsilon$-differential privacy  of
the mechanism  ${\bf \Theta} ({\bf U}, {\bf v}) = {\bf U} + {\bf v}$ 
for 
$\delta$-adjacent pairs
immediately follows
from the Laplace mechanism  for $\epsilon$-differential privacy;
see \cite[Theorem~3.6]{Dwork2014} and 
\cite[Section~5.1]{Weerakkody2019}.

To see that the function $\psi_N$ is decreasing on $(0,\infty)$,
we obtain 
\begin{align}
\psi_N'(y) &=
\frac{y^{2N-1}e^{N-\sqrt{y^2+N^2}}
\big(2N\sqrt{{y^2+N^2}} - y^2 - 2N^2\big)
}{(2N)^{N}  \big(\sqrt{y^2+N^2} - N\big)^{N+1}}. \notag 
\end{align}
Since 
\begin{equation}
\big( 2N\sqrt{{y^2+N^2}} \big)^2 - \big(y^2 + 2N^2\big)^2 = -y^4,
\end{equation}
it follows that $\psi_N'(y) < 0$ for all $y \in (0,\infty)$. Hence
$\psi_N$ is decreasing on $(0,\infty)$.
This completes the proof.
\hspace*{\fill} $\blacksquare$


\end{document}